\documentclass[11pt]{article}
\pdfoutput=1
\usepackage{latexsym, graphicx,cite}
\usepackage{amsmath}
\usepackage{amssymb}
\usepackage{amsthm}
\usepackage{mathtools}
\usepackage{graphicx}
\usepackage{caption}
\usepackage{subcaption}

\usepackage[top = 0.8 in,bottom = 1. in, left = 0.8 in, right=0.8 in]{geometry}

\newcommand{\arXiv}[1]{\href{http://www.arXiv.org/abs/#1}{arXiv:#1}}
\usepackage[colorlinks=true, linkcolor=blue, bookmarks=true]{hyperref}

\makeatletter
\renewcommand\section{\@startsection {section}{1}{\z@}%
	{-1.5ex \@plus -1ex \@minus -.2ex}%nn
	{.3ex \@plus.2ex}%
	{\normalfont\large\bfseries}}
\renewcommand\subsection{\@startsection{subsection}{2}{\z@}%
	{-0.5ex\@plus -1ex \@minus -.2ex}%
	{0.2ex \@plus .2ex}%
	{\normalfont\bfseries}}
\makeatother

\newcommand{\Veff}{V}
\newcommand{\parm}{s}
\newcommand{\Z}{{\mathbb Z}}

\newcommand{\beq}{\begin{equation}}
	\newcommand{\eeq}{\end{equation}}
\newcommand{\beqnn}{\begin{equation*}}
	\newcommand{\eeqnn}{\end{equation*}}

\theoremstyle{proposition}
\newtheorem{proposition}{Proposition}[section]
\theoremstyle{remark}

\theoremstyle{theorem}

\theoremstyle{lemma}
\newtheorem{lemma}{Lemma}[section]

\begin{document}

	\begin{center}
		{\Large\bf Energy cascades  and condensation via coherent dynamics\\ in Hamiltonian systems}
		\vskip 2mm
		{\large Anxo Biasi$^{1}$ and Patrick G\'erard$^{2}$}
		\vskip 1mm
		{\em $^1$ Departamento de F\'isica de Part\'iculas, Universidade de Santiago de Compostela and Instituto Galego de F\'isica de Altas Enerx\'ias (IGFAE), E-15782 Santiago de Compostela, Spain.}
		\vskip -2mm
		{\em $^2$ Universit\'e Paris–Saclay, Laboratoire de Math\'ematiques d'Orsay,\\ CNRS, UMR 8628, F-91405 Orsay, France.}
		\vskip -1mm
		{\small\noindent {\tt $^{1}$ anxo.biasi@gmail.com $\qquad$ $^{2}$ patrick.gerard@universite-paris-saclay.fr}}
		\vskip 0mm
	\end{center}
	\vspace{-0.5cm}
	\begin{center}
		{\bf Abstract}\vspace{-0.6cm}
	\end{center}
	
	\noindent This work makes analytic  progress in the deterministic study of turbulence in Hamiltonian systems by identifying two types of energy cascade solutions and the corresponding large- and small-scale structures they generate. The first cascade represents condensate formation via a highly coherent process recently uncovered, while the second cascade, which has not been previously observed, leads to the formation of other large-scale structures. The concentration of energy at small scales is characterized in both cases by the development of a power-law spectrum in finite time, causing the blow-up of Sobolev norms and the formation of coherent structures at small scales. These structures approach two different types of singularities: a point discontinuity in one case and a cusp in the other. The results are fully analytic and explicit, based on two solvable families of Hamiltonian systems identified in this study.

	%\vfill

%%%%%%%%%%%%%%%%%%%%%%%%%%%%%%%%%%%%%%%%%%%%%%%%%%%%%%%%%%%%%%%%%%%%%%%%%%%%%%%%%%%%%%%%%%%%
%%%%%%%%%%%%%%%%%%%%%%%%%%%%%%%%%%%%%%%%%%%%%%%%%%%%%%%%%%%%%%%%%%%%%%%%%%%%%%%%%%%%%%%%%%%%
%%%%%%%%%%%%%%%%%%%%%%%%%%%%%%%%%%%%%%%%%%%%%%%%%%%%%%%%%%%%%%%%%%%%%%%%%%%%%%%%%%%%%%%%%%%%

	%%%%%%%%%%%%%%%%%%%%%%%%%%%%%%%%%%%%%%%%%%%
	%%%%%%%%%%%%%%%%%%%%%%%%%%%%%%%%%%%%%%%%%%%	
	%%%%%%%%%%%%%%%%%%%%%%%%%%%%%%%%%%%%%%%%%%%
	%%%%%%%%%%%%%%%%%%%%%%%%%%%%%%%%%%%%%%%%%%%
	
	\section{Introduction}
	
	One of the most fascinating features of nonlinear waves is their capability to form structures of vastly different scales. This property has a profound impact on the dynamics of the systems by giving rise to large--- and small---scale elements, such as condensates, solitons, and spikes, which introduce new complexities and behaviors. The formation of these structures is often explained via ``cascades" in momentum space, where energy is transferred toward either high-frequency or low-frequency modes. While this turbulent behavior is reasonably well understood in incoherent regimes, grounded in chaotic dynamics and the wave turbulence theory \cite{NazarenkoBook}, turbulent phenomena  driven by phase-sensitive dynamics (coherent regimes) remain largely unexplored. Despite considerable advances over the past decades \cite{Bourgain,Bourgain2,Kuksin1,Kuksin2,Staffilani2010,GG,Hani2}, predicting when and how cascades happen in coherent regimes is challenging, and further research is needed to classify and understand the phenomena arising from them. Our work makes analytic progress on these questions by presenting an explicit description of cascades driven by highly coherent dynamics and the corresponding structures that emerge from them.
	
	%%%%%%%%%%%%%%%%%%%%%%%%%%%%%%%%%%%%%%%%%%%%%%
	%%%%%%%%%%%%%%%%%%%%%%%%%%%%%%%%%%%%%%%%%%%%%%
	
	\subsection{Fully resonant Hamiltonian systems}
	\label{sec:Setup}
	
	This work focuses on Hamiltonian systems of the form:
	\beq
	i\frac{d\alpha_n}{dt} = \underset{n+m=k+j}{\underbrace{\sum_{m=0}^{\infty}\sum_{k=0}^{\infty}\sum_{j=0}^{\infty}}}C_{nmkj} \bar{\alpha}_m\alpha_k\alpha_j,
	\label{eq:Resonant_Equation}
	\eeq
	where $\alpha_n\in \mathbb{C}$ are complex variables labeled by $n\in\mathbb{N}$, the bar represents complex conjugation, and $C_{nmkj}\in\mathbb{R}$ are time-independent couplings with symmetries: $n\leftrightarrow m$, $k\leftrightarrow j$, and $(n,m)\leftrightarrow(k,j)$. They conserve the ``particle number", the ``energy", and the Hamiltonian:
	\beq
	N = \sum_{n=0}^{\infty} |\alpha_n|^2, \qquad E = \sum_{n=0}^{\infty} n |\alpha_n|^2, \qquad \text{and} \qquad	\mathcal{H} = \frac{1}{2} \underset{n+m=k+j}{\underbrace{\sum_{n=0}^{\infty}\sum_{m=0}^{\infty}\sum_{k=0}^{\infty}\sum_{j=0}^{\infty}}}C_{nmkj} \bar{\alpha}_n\bar{\alpha}_m\alpha_k\alpha_j,
	\label{eq:Conserved_equatities_N_E}
	\eeq
	The denomination for $N$ and $E$ comes from their origin when associated with the nonlinear Schr\"odinger equation \cite{BBCE2}. In addition, these systems enjoy the symmetries: phase-shift ($\alpha_n(t) \to e^{i(\phi + n \theta)} \alpha_n(t)$ with $\phi,\theta\in \mathbb{R}$), scaling ($\alpha_n(t) \to \epsilon \alpha_n(\epsilon^2 t)$ with $\epsilon>0$), and time-reversibility ($\alpha_n(t)\to \bar{\alpha}_n(-t)$).

	The interest in this class of Hamiltonian systems comes from their description of weakly interacting waves and the vast collection of phenomena they display. These systems arise as the resonant approximation of wave models with a fully resonant spectrum of normal frequencies, i.e., $\omega_n = an+b$.  In those cases, $\alpha_n$ represents the amplitude of the $n$th normal mode and the constraint $n+m=k+j$ indicates that resonant interactions, the ones that satisfy  $\omega_n+\omega_m=\omega_k+\omega_j$, dominate. Systems of this form are found in diverse scenarios, such as cold-atoms and nonlinear optics \cite{BBCE1,BBCE2,CMEH}, nonlinear waves on the sphere \cite{CEL} or general relativity \cite{CEV1,Jalmuzna,BEF,MBox,BMR}, among others. For instance, the nonlinear Schr\"odinger equation presents fully resonant spectra when subject to some trapping potentials \cite{BEM2,EquidistantPotentials}. In addition, this class of systems exhibits a rich phenomenology, presenting stationary and time-periodic motions \cite{GGT,GHT,BBCE1,BBCE2,BBE1,CF,BEF,CEL}, breathing modes \cite{E}, Fermi-Pasta-Ulam-Tsingou recurrences \cite{BEM1}, energy cascades \cite{BMR,GG,Xu,BE,Jalmuzna}, or condensation \cite{Biasi}. They have found application in diverse topics such as vortex motion in Bose-Einstein condensates \cite{BBCE1}, formation of tiny black holes in general relativity \cite{BMR}, or the increase of Sobolev norms in Hamiltonian PDEs \cite{Survey_Szego}. All of it makes Hamiltonian systems of the form (\ref{eq:Resonant_Equation}) an interesting source of elaborate phenomena with application in multiple fields. We therefore expect that our results will lead to better compression of energy cascades and structure formation in a broad class of topics.
	
	%%%%%%%%%%%%%%%%%%%%%%%%%%%%%%%%%%%%%%%%%%%%%%%%%%%%%%%%%%%%%%%%%%%%%%%%%
	
	\subsection{Energy cascades}
	\label{subsec:Energy cascades}
	This work characterizes processes of energy transfer from low modes (small $n$) to arbitrarily high ones (large $n$). As later discussed, these processes represent energy concentration in arbitrarily small regions, formation of sharp structures, and loss of regularity. We focus on scenarios where the amount of energy at small scales is strongly suppressed (the amplitude spectrum $|\alpha_n|^2$ decays exponentially for large $n$)  but weakens over time, eventually giving rise to a power-law (asymptotic) spectrum in either finite or infinite time $T$:
	 \beq
	 |\alpha_{n\gg1}(t)|^2\ \sim\ |c(t)|^2 n^{\gamma} \left(\frac{x(t)}{x_c}\right)^n\ \underset{t\to T}{\longrightarrow} \ |c(T)|^2 n^{\gamma},
	 \label{eq:power-law_definition}
	 \eeq
	 with $\gamma\in \mathbb{R}$, $c\in \mathbb{C}$, and $x_c>x\geq0$,  but with $x(t) \to x_c$ as $t\to T$. We label these processes by the power-law exponent $\gamma$, referring to them as {\em $\gamma$-cascades}. In formal terms, these processes of energy migration are quantified by the growth of Sobolev norms:
	\beq
		H^{\xi} = \left(\sum_{n=0}^{\infty}(n+1)^{2\xi}|\alpha_n|^2\right)^{1/2}.
		\label{eq:Sobolev_norms_definition}
	\eeq
	Starting from conditions where all these norms are finite, the question is whether some of them can become unbounded in finite or infinite time, how this happens, and which of them remain finite.
	 
	In fully resonant systems (\ref{eq:Resonant_Equation}), the unbounded growth of Sobolev norms/energy cascades has been only characterized in a few works. From an analytic perspective, most of the results have been obtained around the {\em cubic Szeg\H{o} equation} introduced by one of the present authors in Ref.~\cite{GG}; see Ref.~\cite{Survey_Szego} for a survey. The first studies of this model revealed explicit solutions where Sobolev norms grew as much as desired but remained bounded (the system almost developed a $0$-cascade), while it was later shown that the model admits a superpolynomial growth of these norms. Modifications of the cubic Szeg\H{o} equation led to explicit examples of a $0$-cascade, with exponential growth of Sobolev norms in Ref.~\cite{Xu} ({\em the $\alpha$-Szeg\H{o} equation}) and both exponential and polynomial growth in Ref.~\cite{BE} ({\em $\beta$-Szeg\H{o} equation}). In Refs.~\cite{GG2,GGH} ({\em the damped Szeg\H{o} equation}), it was shown that a damping term enhanced the unbounded increase of Sobolev norms, likewise, other modifications presented in Refs.~\cite{Pocovnicu,Thirouin}. In all these results, Sobolev norms displayed unbounded increase but remained finite (i.e., they exhibited infinite time blow-up). Finite time blow-up was proved for the first time in a work by one of the current authors in Ref.~\cite{Biasi}, with the display of a $-3/2$-cascade. From a numerical perspective, the work in Ref.~\cite{BR} drew much attention to the connection between energy cascades and black hole formation in Einstein equations, triggering a surge of research; see Refs.~\cite{BMR, MBox, Jalmuzna} for specific contributions  and Ref.~\cite{E2} for a review.  See also Refs.~\cite{Bourgain,Staffilani2010,Guardia,Maspero,Hani,Hani2,Thomann,GP,GL,KKK,Maspero2,Camps} and references therein for results on the growth of Sobolev norms beyond the class of fully resonant systems. 
	
	%%%%%%%%%%%%%%%%%%%%%%%%%%%%%%%%%%%%%%%%%%%%%%
	%%%%%%%%%%%%%%%%%%%%%%%%%%%%%%%%%%%%%%%%%%%%%%
	
	\subsection{Condensation}
	\label{sec:Condensation_Intro}

	The phenomenon of condensation is an important element in our study, so we introduce it here.

	The most familiar manifestation of condensation is Bose–Einstein condensation (BEC) in quantum mechanics \cite{Pitaevskii}, where bosons occupy the ground state of the system at sufficiently low temperature, forming a macroscopic coherent object. Analogous behavior has been observed in classical nonlinear waves \cite{Condensation_2004,Condensation_2005,Condensation_2018,Condensation_2020,Condensation_2012} with at least two conserved quantities---commonly referred to as the ``particle number" and the energy. In this context, condensation refers to the accumulation of particle number at the lowest-energy mode, such that a significant fraction is supported by that mode. This phenomenon has been studied extensively in the physics literature of weakly nonlinear random waves under wave turbulence theory \cite{Condensation_2001, Condensation_2004,Condensation_2005,NazarenkoBook,Condensation_2011,Condensation_2019}, though a rigorous mathematical description remains an ongoing problem \cite{Escobedo,StaffilaniTran}.

	Our work is concerned with a distinct process of condensation, called coherent condensation, first introduced in \cite{Biasi}. It entails the accumulation of the particle number $N$ (see \eqref{eq:Conserved_equatities_N_E}) to the lowest-energy mode. Formally, this corresponds to the convergence of the amplitude spectrum to a Kronecker delta distribution centered at mode $n=0$,
	\beq
		|\alpha_n(t)|^2 \underset{t\to T}{\longrightarrow} N \delta_{0,n}.
		\label{eq:Condensation_v0}
	\eeq
	The main distinction between coherent condensation and previous observations lies in the nature of the dynamics. In wave turbulence, condensation admits a statistical description, relying on chaotic dynamics and statistically independent mode phases. By contrast, coherent condensation proceeds deterministically: the phases of $\alpha_n(t)$ align along a straight line for all modes except the lowest one. Therefore, this organized behavior stands in sharp contrast to the traditional condensation processes described in the literature.

	Note that the spectrum (\ref{eq:Condensation_v0}) represents the ground state of the Hamiltonian system (\ref{eq:Resonant_Equation}), with energy $E=0$. Since both $N$ and $E$ are conserved, the evolution should not approach this spectrum from initial conditions with $E>0$. However, we shall see that it happens, in certain sense, thanks to the display of an energy cascade that accompanies the condensation.

	%%%%%%%%%%%%%%%%%%%%%%%%%%%%%%%%%%%%%%%%%%%%%%
	%%%%%%%%%%%%%%%%%%%%%%%%%%%%%%%%%%%%%%%%%%%%%%
	%%%%%%%%%%%%%%%%%%%%%%%%%%%%%%%%%%%%%%%%%%%%%%
	%%%%%%%%%%%%%%%%%%%%%%%%%%%%%%%%%%%%%%%%%%%%%%

	\section{Main results}	
	\label{sec:Main_Results}
	
	This work presents analytic progress on the problems of energy cascades and structure formation in Hamiltonian systems (\ref{eq:Resonant_Equation}). Before, only the system reported in Ref.~\cite{Biasi} was confirmed to display cascades in finite time and we now report two infinite families of systems exhibiting the phenomenon. One of these families admits explicit solutions representing a $-3/2$-cascade in finite time and the dynamical formation of a condensate. The second family displays a $-5/2$-cascade in finite time, which has not been previously observed in the literature. The overall aspects of these results are described below, while a detailed analysis is given in sections~\ref{sec:Z_Hamiltonian_systems} and \ref{sec:Y_Hamiltonian_systems}.

	%%%%%%%%%%%%%%%%%%%%%%%%%%%%%%%%%%%%%%%%%%%%%%%%%%%%%%%%
	%%%%%%%%%%%%%%%%%%%%%%%%%%%%%%%%%%%%%%%%%%%%%%%%%%%%%%%%
	
	\subsection{Results on $\mathbf{-3/2}$-cascades and coherent condensation} 
	
	In section~\ref{sec:Z_Hamiltonian_systems}, we present an infinite family of Hamiltonian systems that admit $-3/2$-cascade solutions. We report full description of these processes and show that they accompany the phenomenon of {\em coherent condensation} introduced in Ref.~\cite{Biasi} and above in section~\ref{sec:Condensation_Intro}.
	
	Fig.~\ref{fig:All_Type_1} illustrates the condensation process described by one of the solutions presented in section ~\ref{sec:Z_Hamiltonian_systems}. We observe the convergence of the amplitude spectrum to the lowest mode: $|\alpha_n|^2 \to N \delta_{0,n}$, the development of the power law: $|\alpha_{n\gg1}|^2\sim n^{-3/2}$, and the separation of the quantities $N$ and $E$ in the spectrum. While $N$ experiences an inverse transfer to end up entirely stored at the lowest mode, $E$ experiences a direct transfer to be stored at arbitrarily high modes. This dual-cascade behavior causes the formation of structures of dramatically different scale in position space. Fig.~\ref{fig:All_Type_1}(e) illustrates it in a one-dimensional periodic box by using $u(t,\theta)= \sum_{n=0}^{\infty}\alpha_n(t) e^{i\theta}$ with $\theta \in [0,2\pi)$. The figure shows the development of a large-scale structure associated with the concentration of $N$ at the lowest mode: a condensate (flat profile), and a small-scale coherent structure associated with the transfer of $E$ to arbitrarily small regions: a ``spike" of finite amplitude but whose width goes to zero, approaching a point discontinuity as $t\to T$. 
	
	Regarding Sobolev norms, the ones with $\xi>1/2$ grow to infinity as the formation of the condensate at time $T$ approaches:
	\beq
	H^{\xi>1/2} = \left(\sum_{n=0}^{\infty}(n+1)^{2\xi}|\alpha_n|^2\right)^{1/2} \underset{t\sim T}{\sim} (T-t)^{2(1-2\xi)}.
	\label{eq:Sobolev_norms_Type_I}
	\eeq
	
	\begin{figure}[h!]  %This figure should be alone on the page, if possible
		\centering
		\includegraphics[width = 8.5cm]{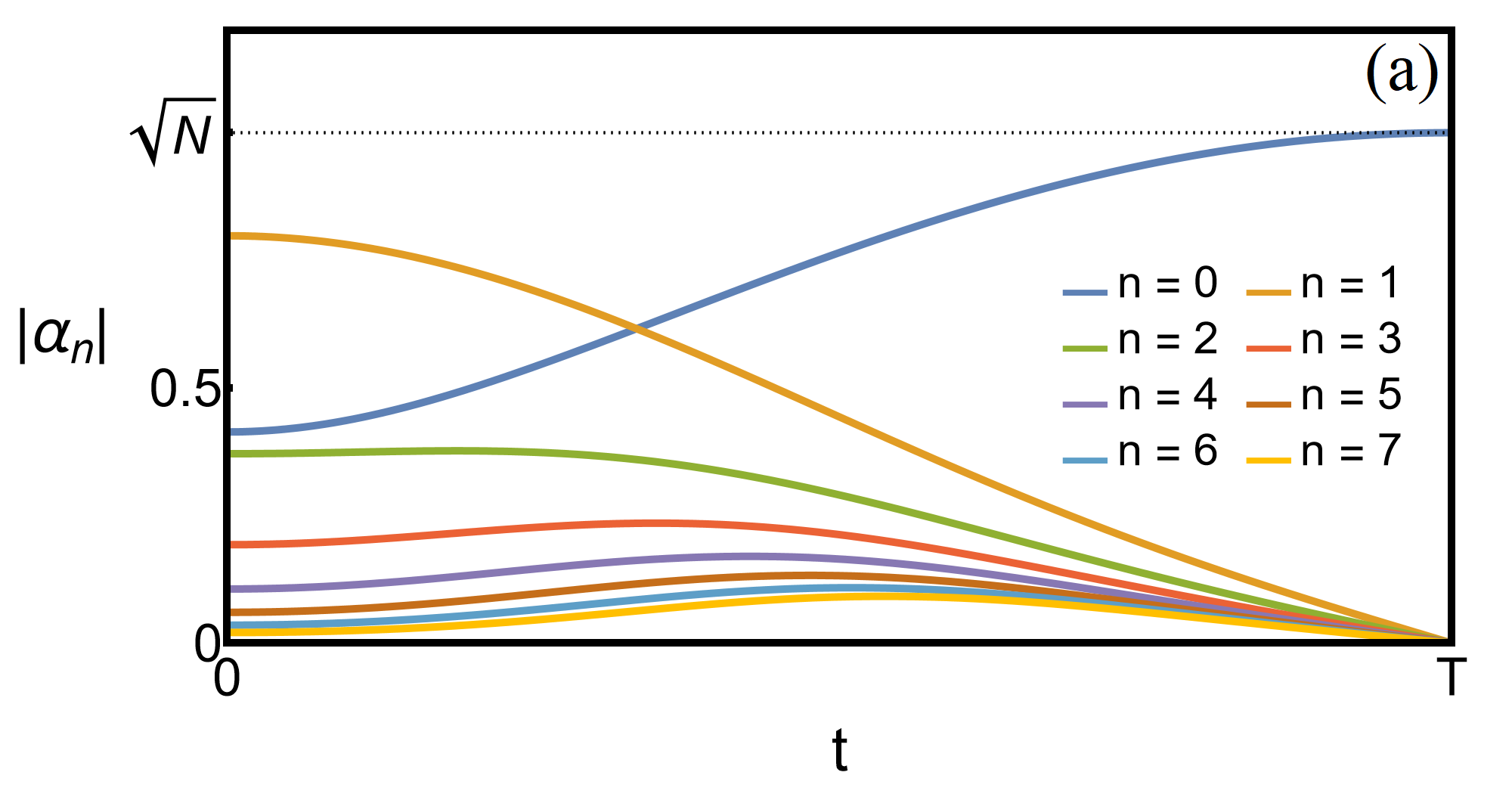}
		\hspace{0.cm}
		\includegraphics[width = 8.5cm]{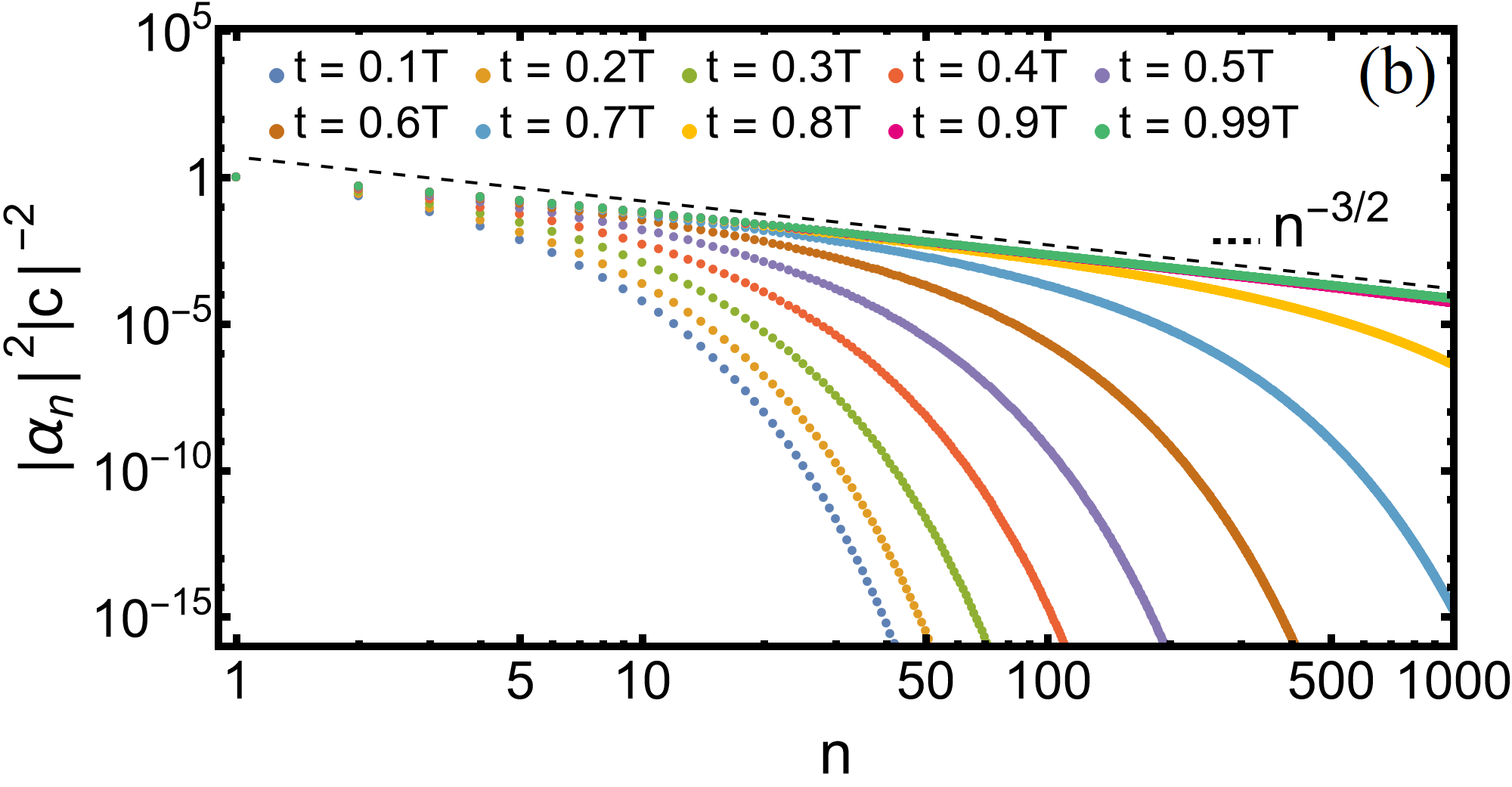}
		
		\vspace{0.5cm}
		\includegraphics[width = 8.5cm]{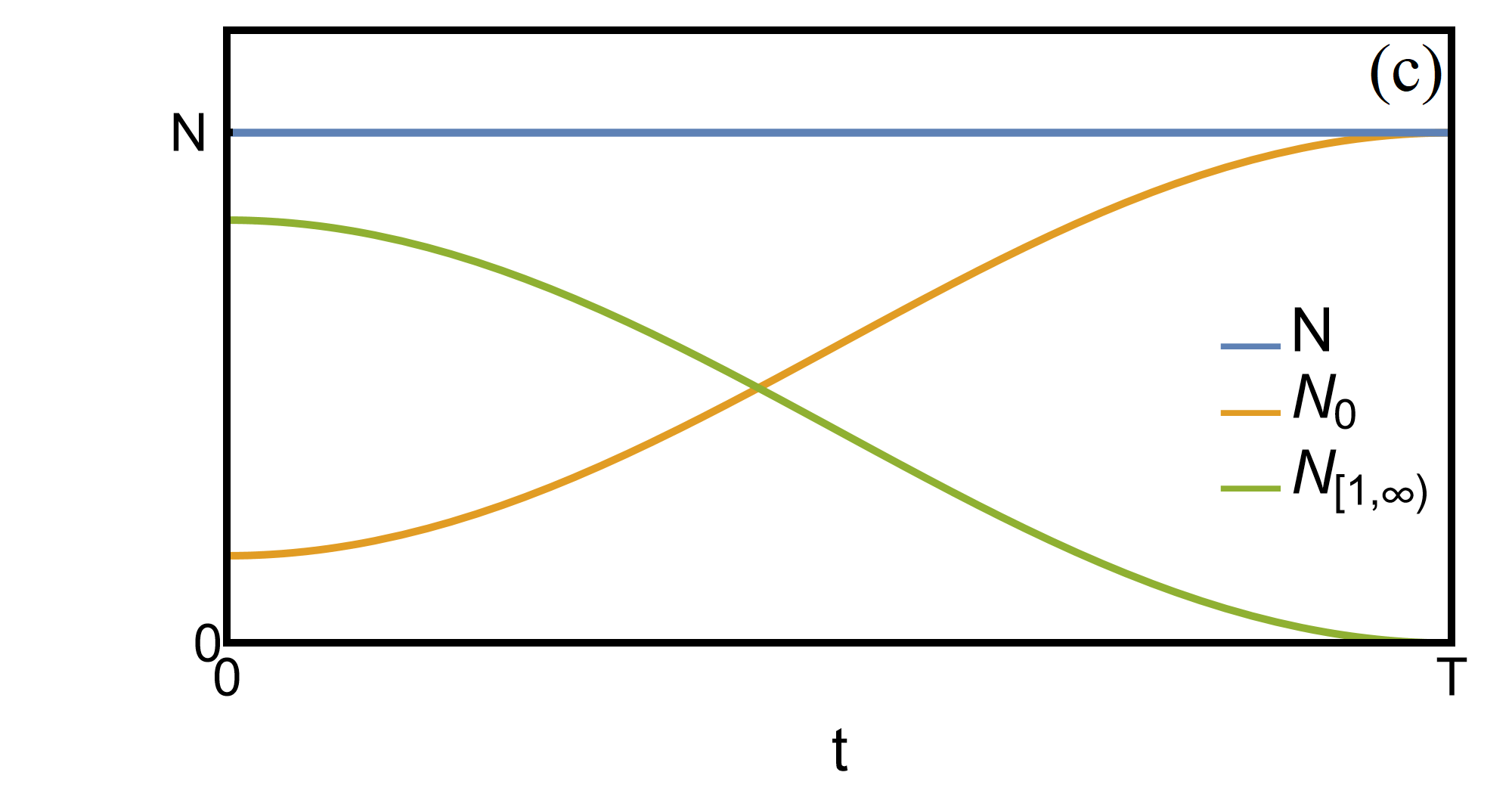} \hspace{0.cm}
		\includegraphics[width = 8.5cm]{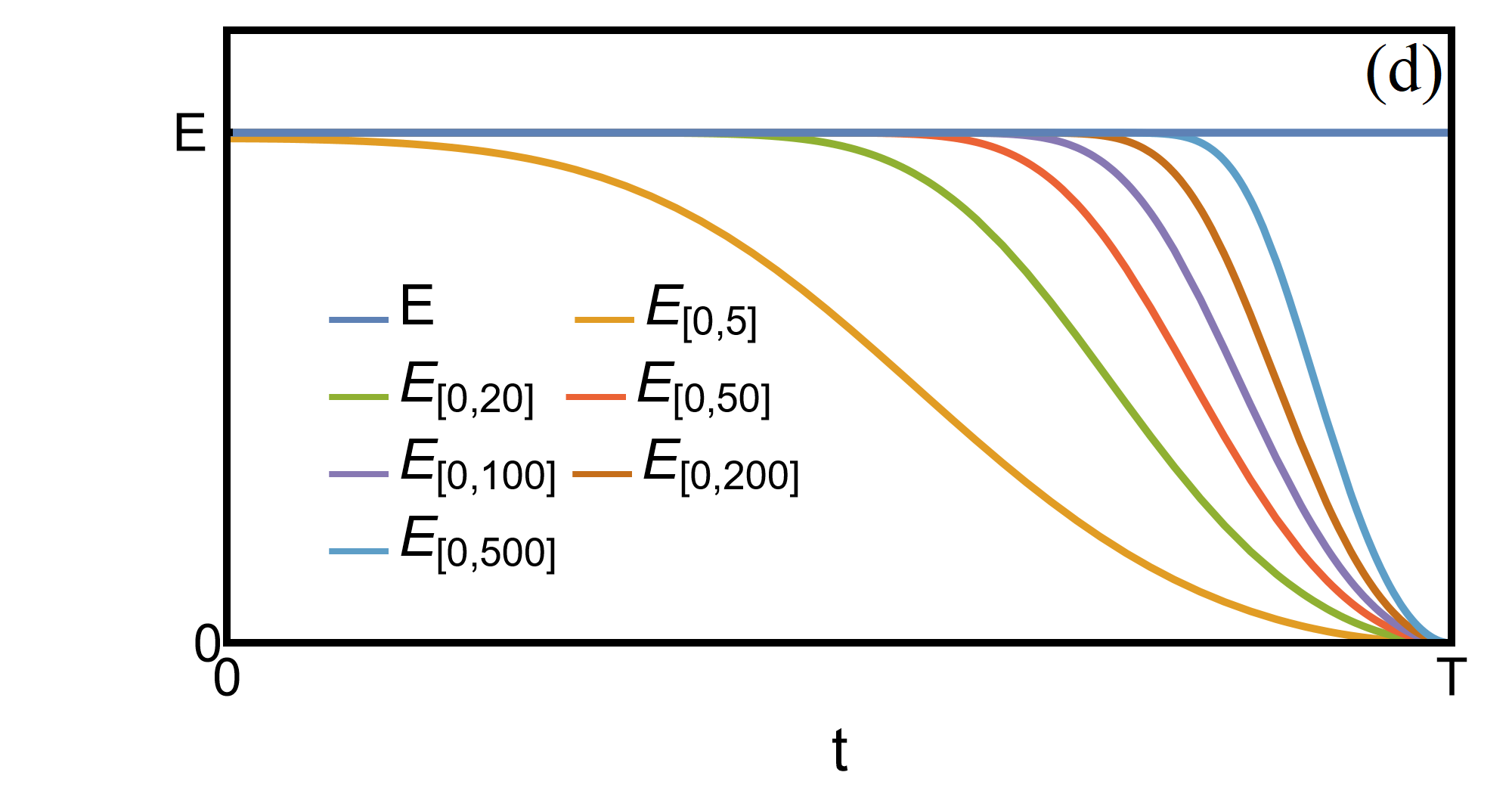}
		
		\vspace{0.5cm}
		\includegraphics[width = \textwidth]{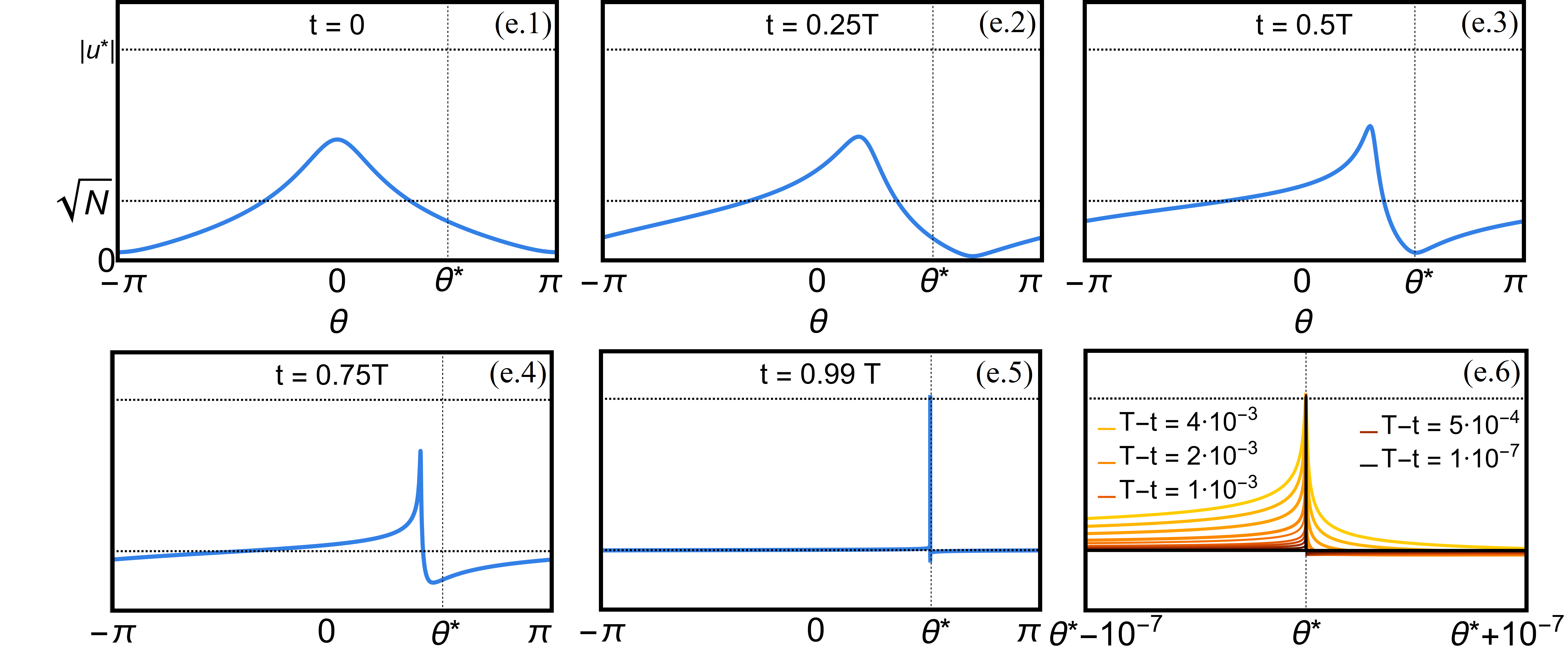}
		\caption{\small Analytic solution representing a $-3/2$-cascade in finite time $T$, and specifically, a process of coherent condensation. (a): Evolution of the first elements of the amplitude spectrum, observing $|\alpha_n(t)|^2 \to N \delta_{0,n}$. (b): Development of the power law $|\alpha_{n\gg 1}|^2\sim |c|^2 n^{-3/2}$, where $|c|^2$ decays to zero. (c, d): Time-evolution of the amount of conserved quantities, $N$ (c) and $E$ (d), stored at the range of modes indicated in the subscript. (e): Position space representation of the condensation process through $|u(t,\theta)|$ at several times. The spike-like structure has finite amplitude but narrows to a point $\theta^*$. Plot (e.6) shows a more detailed view around that point.}
		\label{fig:All_Type_1}
	\end{figure}

	%%%%%%%%%%%%%%%%%%%%%%%%%%%%%%%%%%%%%%%%%%%%%%%%%%%%
	%%%%%%%%%%%%%%%%%%%%%%%%%%%%%%%%%%%%%%%%%%%%%%%%%%%%
	
	\subsection{Results on $\mathbf{-5/2}$-cascades and structure formation}

	Section~\ref{sec:Y_Hamiltonian_systems} provides a family of systems that admit explicit solutions representing $-5/2$-cascades in finite time. To our knowledge, this is the first observation of this type of cascade in the class of Hamiltonian systems (\ref{eq:Resonant_Equation}). Interestingly, they represent dynamics markedly different from the coherent condensation displayed by our $-3/2$-cascade solutions, as appreciated by comparing Fig.~\ref{fig:All_Type_2} with Fig.~\ref{fig:All_Type_1}. This new figure shows that all modes remain excited when the power law is developed and there is no sharp separation of the quantities $N$ and $E$ in the spectrum. Only fractions of them reach the lowest mode and arbitrarily high ones. These differences are also manifested in position space and the behavior of Sobolev norms. Placing the process in a one-dimensional periodic box, $\sum_{n=0}^{\infty}\alpha_n(t) e^{i\theta}$, Fig.~\ref{fig:All_Type_2}(e) shows that $|u(t,\theta)|^2$ does not converge to a flat profile with a narrowing spike. Instead, it has a dependence on $\theta$ at large scales, and develops a different kind of small-scale structure: a sharpening ``pointed corner", which approaches a cusp as $t\to T$. Regarding Sobolev norms, only the ones with $\xi \geq 3/4$ diverge as $t\to T$:
	\beq
		H^{\xi>3/4} \underset{t\sim T}{\sim} (T-t)^{3-4\xi}, \qquad \text{and} \qquad H^{3/4} \underset{t\sim T}{\sim} \ln\left(\frac{1}{T-t}\right).
	\label{eq:Sobolev_norms_Type_II}
	\eeq

	\begin{figure}[h!] %This figure should be alone on the page, if possible
		\centering
		
		\vspace{0.7cm}
		\includegraphics[width = 8.5cm]{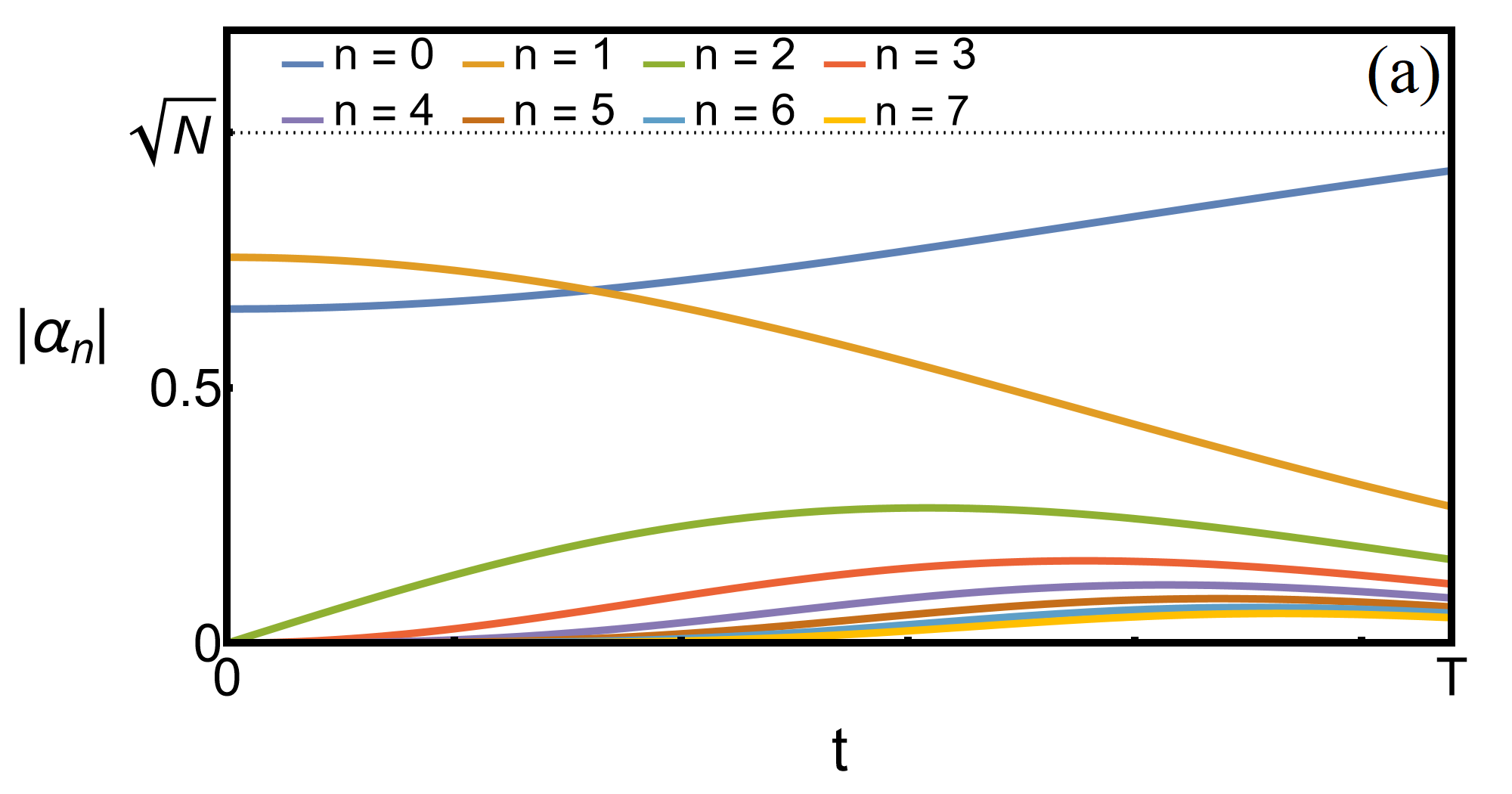} \hspace{0.cm}
		\includegraphics[width = 8.5cm]{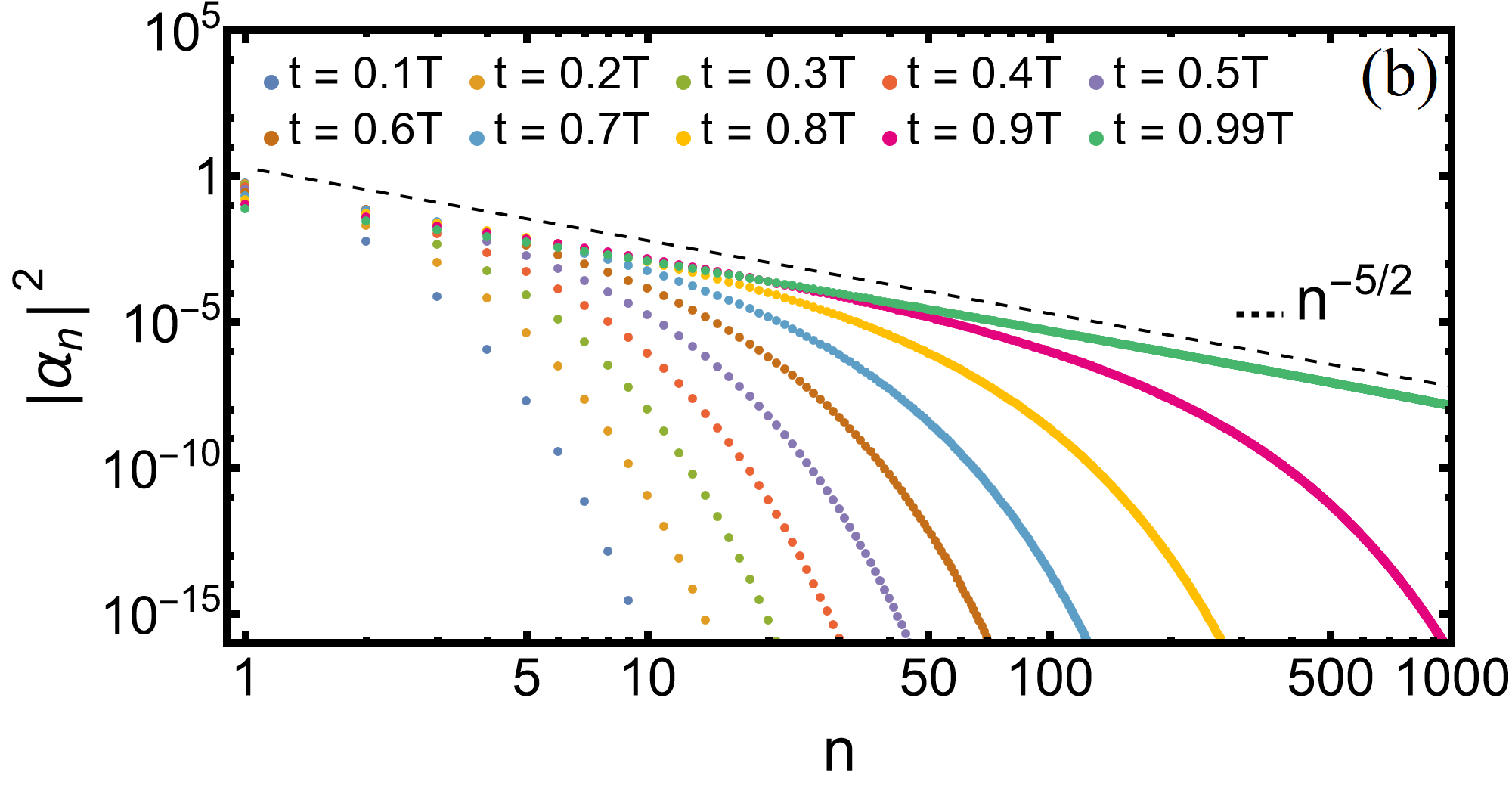}
		
		\vspace{0.99cm}
		\includegraphics[width = 8.5cm]{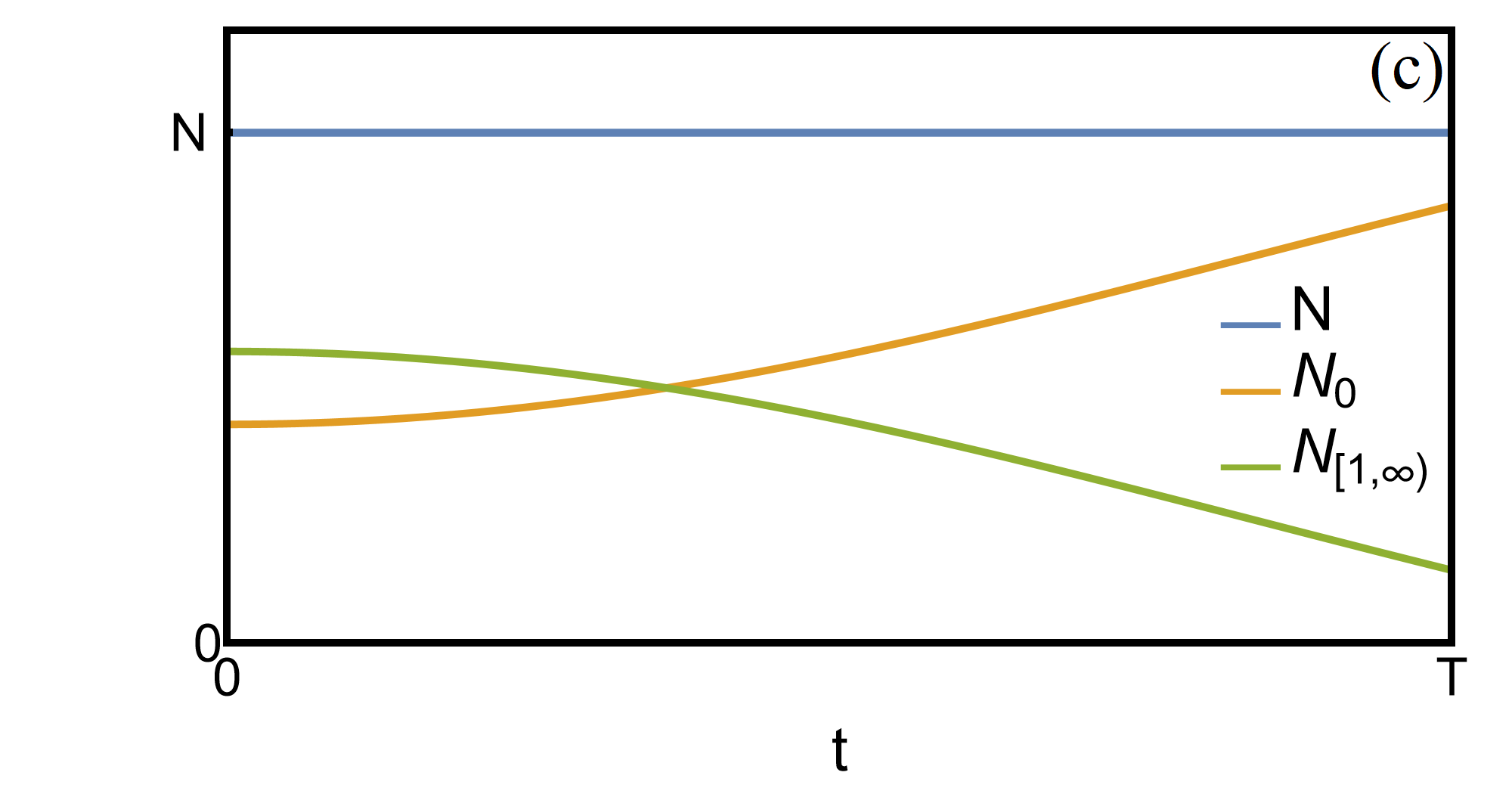} \hspace{0.cm}
		\includegraphics[width = 8.5cm]{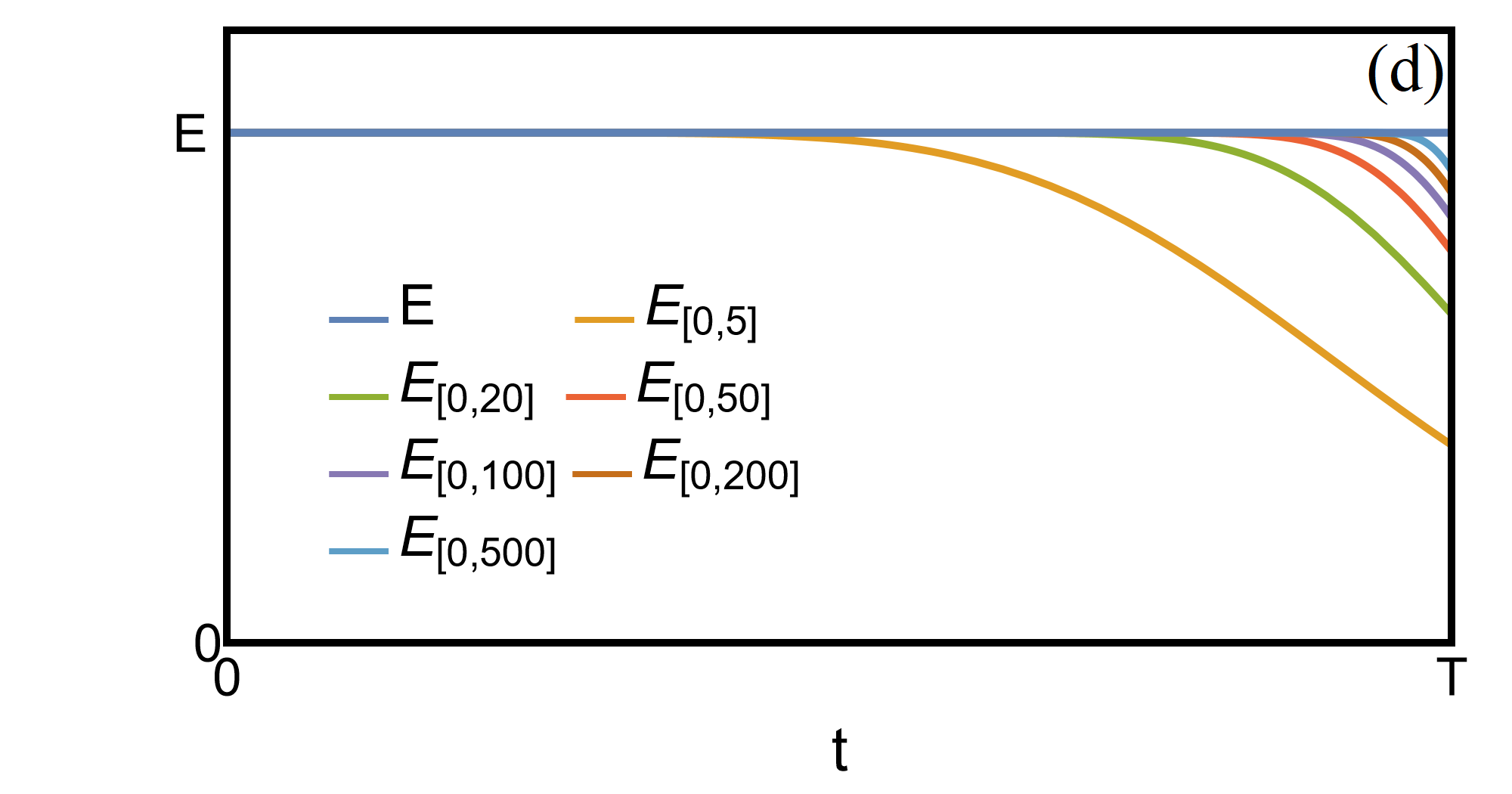}
		
		\vspace{0.99cm}
				\centering
		\includegraphics[width = \textwidth]{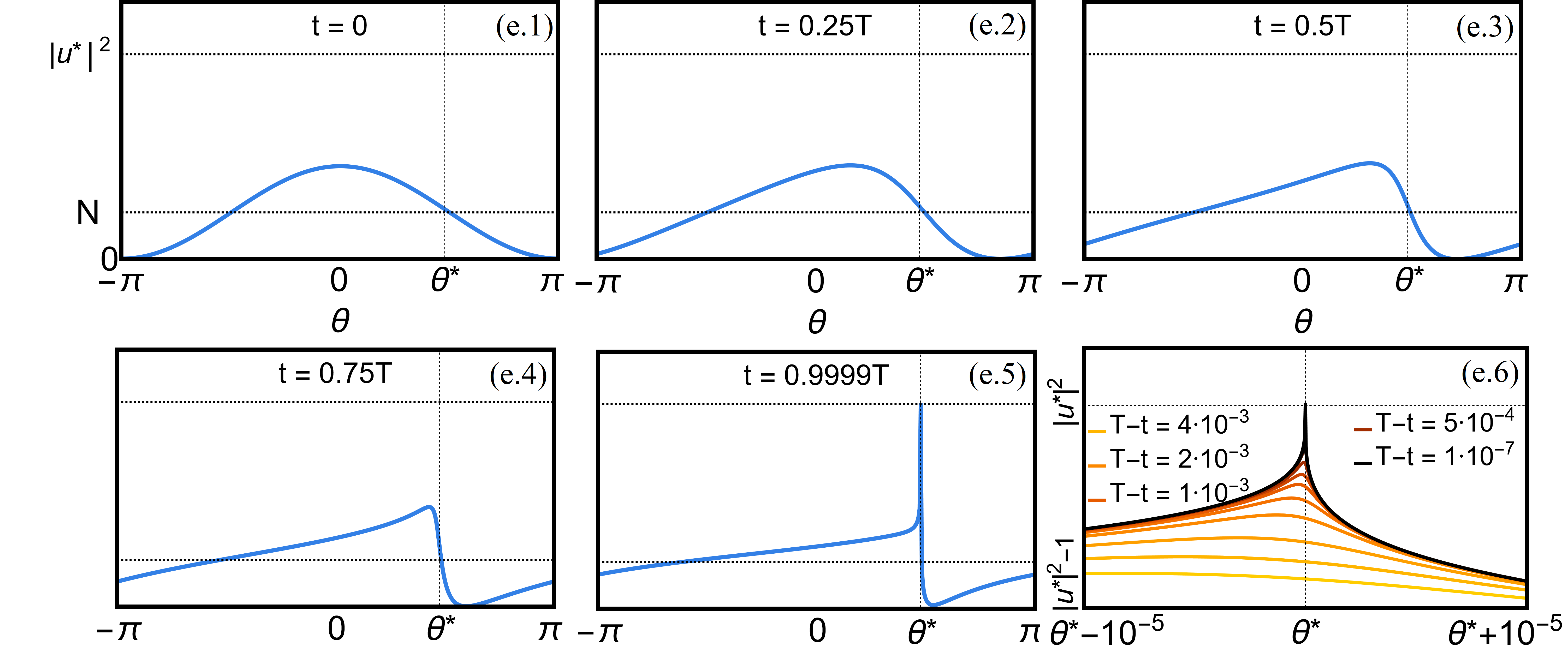}

		\caption{\small Analytic solution representing a $-5/2$-cascade in finite time $T$. (a): Evolution of the first elements of the amplitude spectrum. (b): Development of the power law $|\alpha_{n\gg 1}|^2\sim n^{-5/2}$. (c, d): Time-evolution of the amount of $N$ (c) and $E$ (d) stored at the range of modes indicated in the subscript. (e): Position space representation $|u(t,\theta)|^2$ at several times. A sharpening pointed corner arises at a point $\theta^*$. Plot (e.6) shows a more detailed view around that point. Note the different nature between the singularities emerging here and in Fig.~\ref{fig:All_Type_1}(e.6). \vspace{0.5cm}}
		\label{fig:All_Type_2}
	\end{figure}

	%%%%%%%%%%%%%%%%%%%%%%%%%%%%%%%%%%%%%%%%%%%%%%%%%%%
	%%%%%%%%%%%%%%%%%%%%%%%%%%%%%%%%%%%%%%%%%%%%%%%%%%%
	%%%%%%%%%%%%%%%%%%%%%%%%%%%%%%%%%%%%%%%%%%%%%%%%%%%
	%%%%%%%%%%%%%%%%%%%%%%%%%%%%%%%%%%%%%%%%%%%%%%%%%%%

	\section{The Z-Hamiltonian systems: $\mathbf{-3/2}$-cascades \& coherent condensation}
	\label{sec:Z_Hamiltonian_systems}
	
	This section presents our results on $-3/2$-cascades. They are centered on the following family of Hamiltonian systems we have found. They are named {\em Z-Hamiltonians} and are given by the couplings:
	\beq
	C_{nmkj}^{(\parm)} = \left( \frac{\parm-1}{2}(n+m)+1\right) \frac{f_n f_m f_k f_j}{f_{n+m}^2},
	\label{eq:C_nmij_equation}
	\eeq
	where $f_n=\sqrt{A_n^{(\parm)}}$ and $A_n^{(\parm)}$ are Fuss-Catalan numbers \cite{FussCatalan} parameterized by $\parm\geq 1$
	\beq
	A_n^{(\parm)} = \frac{\Gamma(\parm n+1)}{\Gamma((\parm-1)n+2)\Gamma(n+1)}.
	\label{eq:Fuss_Catalan_numbers}
	\eeq
	The form of these coefficients is peculiar but crucial for our analysis. It guarantees the existence of an invariant manifold that we exploit in section~\ref{sec:invariant_manifold} to study energy cascades and condensation processes. This family of Hamiltonian structures was constructed through a trial-and-error process, aimed at generalizing two previously known systems sharing a similar invariant manifold: the model exhibiting coherent condensation \cite{Biasi} with $\parm=2$, and the well-studied cubic Szeg\H{o} equation \cite{GG} with $\parm=1$. Most systems of this kind \cite{BE,Xu,GGH,GG2} have been developed by such ad hoc procedures. However, we are now developing a fully constructive method that will be presented in forthcoming work.

		We here discuss the general case $\parm>1$ and make connections with the results for $\parm=1$ in section~\ref{sec:The cubic Szego equation}. However, before proceeding further, it is worth mentioning that the Z-systems belong to the same class of Hamiltonian structures observed in physical models, in  the sense that the coefficients $C_{nmkj}$ exhibit, at most, polynomial growth for large indices,  such as found for nonlinear Schr\"odinger equations \cite{BMP} or relativistic wave equations \cite{CEV2}. Therefore, we expect to observe  dynamics exhibited by the Z-Hamiltonians in physical scenarios. In this sense, our systems demonstrate the phenomena to look for and provide insight into the mechanisms that govern them. The polynomial growth of the coefficients comes by using the resonance condition $n+m=k+j$ and the asymptotic growth 
	\beq
	f_{n\gg1} \sim \left(\frac{\parm}{2\pi (\parm-1)^3}\right)^{1/4} n^{-3/4} x_c^{-n/2},
	\label{eq:fgrowth}
	\eeq 
	where $x_c = (\parm-1)^{\parm-1}/\parm^\parm < 1$. This latter quantity will be of great importance in subsequent sections.

	%%%%%%%%%%%%%%%%%%%%%%%%%%%%%%%%%%%%%%%%%%%%%%
	%%%%%%%%%%%%%%%%%%%%%%%%%%%%%%%%%%%%%%%%%%%%%%
	%%%%%%%%%%%%%%%%%%%%%%%%%%%%%%%%%%%%%%%%%%%%%%
	%%%%%%%%%%%%%%%%%%%%%%%%%%%%%%%%%%%%%%%%%%%%%%
	
	\subsection{Local well-posedness}
	\label{sec:Local_well_posedness}
	
	Before entering into the phenomenon of coherent condensation, we settle the Z-Hamiltonian systems on solid mathematical ground by proving their local well-posedness.

\begin{proposition}\label{th:WPZ}
Let $\xi >3/2$ and $(\alpha_{n,0})_{n\geq 0}$ be a sequence of complex numbers such that
\beq \sum_{n=0}^\infty (1+n)^{2\xi}\vert \alpha_{n,0}\vert ^2 \leq R<\infty \ . \label{eq:Zdata}\eeq
There exists $\tau =\tau (R)>0$ and a unique sequence $(\alpha_n(t))_{n\geq 0}$ of $C^1$ functions 
on the interval $[-\tau ,\tau ]$ satisfying $\alpha_n(0)=\alpha_{n,0}$, the estimate 
\beq
\sup_{|t|\leq \tau } \sum_{n=0}^\infty (1+n)^{2\xi}|\alpha_n(t)|^2<\infty ,
\label{eq:Zbounds}
\eeq
and Equation \eqref{eq:Resonant_Equation} with \eqref{eq:C_nmij_equation}. Furthermore, the conservation laws \eqref{eq:Conserved_equatities_N_E}
hold for this solution.
\end{proposition}
\begin{proof}
We will make an extensive use of the following lemma, which is a simple consequence of the Cauchy--Schwarz inequality.
\begin{lemma}\label{lem:sums}
Let $(a_k)_{k\in \Z}, (b_m)_{m\in \Z}, (c_n)_{n\in \Z},(d_j)_{j\in \Z},$ be $\Z$--indexed sequences of nonnegative numbers.The following estimates hold.
\begin{align}
&     \sum_{k+m+n=0}a_kb_mc_n\leq \left (\sum_k a_k^2\right )^{1/2} \left (\sum_m b_m^2\right )^{1/2} \left(\sum_n c_n\right ),\label{eq:3sum}\\
&     \sum_{k+m+n=0}a_kb_mc_n d_j\leq \left (\sum_k a_k^2\right )^{1/2} \left (\sum_m b_m^2\right )^{1/2} \left(\sum_n c_n\right )\left(\sum_j d_j\right ),\label{eq:4sum}
\end{align}
\end{lemma}
Let us come to the proof of Proposition \ref{th:WPZ}. For every positive integer $L$, we define 
$$\alpha_{n,0}^{(L)}:= \alpha_{n,0}{\bf 1}_{[0,L]}(n)\ ,$$
and we consider the following cutoff approximation of \eqref{eq:Resonant_Equation}, 
\beq i\dot \alpha_n^{(L)}=\sum_{\substack{0\leq j,k,m\leq L \\ j+k=m+n}}C_{nmkj}\alpha_j^{(L)}\alpha_k^{(L)}\overline \alpha_m^{(L)}\ ,\ \alpha_n^{(L)}(0)=\alpha_{n,0}^{(L)}\ .
\label{eq:Zapprox}
\eeq 
The Cauchy--Lipschitz theorem ensures the existence of a unique solution to this ODE on a small time interval. Furthermore, the symmetries of the coefficients $C_{nmkj}$ yield the following mass conservation law,
\beq
\sum_{n=0}^L |\alpha_n^{(L)}(t)|^2=\sum_{n=0}^L |\alpha_{n,0}^{(L)}|^2\ ,
\label{eq:mass}
\eeq 
from which we infer the existence of the solution for all times. We now claim the following key a priori estimate, for some $\tau =\tau(R)>0$ to be chosen,
\beq \sup_{|t|\leq \tau }\sum_{n=0}^L n^{2\xi} |\alpha_n^{(L)}(t)|^2 \leq 2R\ .
\label{eq:Zest}
\eeq 
Let us prove \eqref{eq:Zest}. Set $$H^\xi _L(t):=\sum_{n=0}^L n^{2\xi}|\alpha_n^{(L)}(t)|^2\ .$$
Then, using again the symmetries of $C_{nmkj}$, we have
\begin{align*}
\frac{dH^\xi_L}{dt}&=2{\rm Im}\sum_{\substack{0\leq j,k,m,n\leq L\\ j+k=m+n}}n^{2\xi}C_{nmkj}\alpha_j^{(L)}(t)\alpha_k^{(L)}(t)\overline \alpha_m^{(L)}(t)\overline \alpha_n^{(L)}(t)\\
&=\frac{1}{2i}\sum_{\substack{0\leq j,k,m,n\leq L\\ j+k=m+n}}\frac{f_nf_mf_kf_j}{f_{n+m}f_{k+j}}\left (\frac{\parm-1}{2}(n+m)+1\right )(n^{2\xi}+m^{2\xi}-k^{2\xi}-j^{2\xi})\alpha_j^{(L)}(t)\alpha_k^{(L)}(t)\overline \alpha_m^{(L)}(t)\overline \alpha_n^{(L)}(t)\ .
\end{align*}
Notice that, from the asymptotic estimate \eqref{eq:fgrowth}, we know that the quantity $f_nf_m/f_{n+m}$ is bounded. On the other hand, we have the following elementary estimate,
$$(n+m)|n^{2\xi}+m^{2\xi}-k^{2\xi}-j^{2\xi}|\lesssim \max (m,n)^\xi \max (j,k)^\xi (\min (j,k)+\min (m,n))\ .$$
Indeed, assuming e.g. $m\leq n$ and $j\leq k$, and using $|n-k|=|j-m|, k=n+m-j\leq 2n, n=j+k-m\leq 2k$, we obtain
\begin{align*}
(n+m)|n^{2\xi}+m^{2\xi}-k^{2\xi}-j^{2\xi}|&\lesssim n|j-m|(n^{2\xi -1}k^{2\xi -1})\\
&\lesssim (j+m)n^\xi k^\xi \ .
\end{align*}
Consequently, using \eqref{eq:4sum} and \eqref{eq:mass}, we infer
\begin{align*}
\left |\frac{dH^\xi_L}{dt}\right |&\lesssim \sum_{\substack{0\leq j,k,m,n\leq L\\ j+k=m+n}}(j+m)n^\xi k^\xi |\alpha_j^{(L)}(t)||\alpha_k^{(L)}(t)||\alpha_m^{(L)}(t)||\alpha_n^{(L)}(t)|\\
&\lesssim H^\xi _L(t)^2\ ,
\end{align*} 
where we have used, for every $\xi >3/2$,
$\displaystyle{\sum_{n=0}^L n|\alpha_n^{(L)}(t)|\leq B(\xi ) H^\xi_L(t)^{1/2}}\ .$
Since $H^\xi_L (0)$ is uniformly bounded by $R$ in view of \eqref{eq:Zdata}, we get $H^\xi _L(t)\leq 2R$ for $|t|\leq \tau (R)$ small enough.\\
As a final step of the proof, we grab some contraction estimate, by writing
\begin{align*}
i\frac d{dt}(\alpha_n^{(L+1)}(t)-\alpha_n^{(L)}(t))&=\sum_{\substack{0\leq j,k,m\leq L\\ j+k=m+n}}C_{nmkj}(\alpha_j^{(L+1)}(t)-\alpha_j^{(L)}(t))\alpha_k^{(L+1)}(t)\overline \alpha_m^{(L+1)}(t)\\
&+ \sum_{\substack{0\leq j,k,m\leq L\\ j+k=m+n}}C_{nmkj}\alpha_j^{(L)}(t))(\alpha_k^{(L+1)}(t)-\alpha_k^{(L)}(t))\overline \alpha_m^{(L+1)}(t)  \\
&+\sum_{\substack{0\leq j,k,m\leq L\\ j+k=m+n}}C_{nmkj}\alpha_j^{(L)}(t)\alpha_k^{(L)}(t)(\overline \alpha_m^{(L+1)}(t)-\overline \alpha_m^{(L)}(t))\\
&+\sum_{\substack{0\leq j,k,m\leq L+1,\\ L+1\in \{j,k,m\},\\ j+k=m+n}}C_{nmkj}\alpha_j^{(L+1)}(t)\alpha_k^{(L+1)}(t)\overline \alpha_m^{(L+1)}(t) .
\end{align*}
Using the bound \eqref{eq:Zest}, we infer, for $|t|\leq \tau$,
$$
\left | \frac d{dt}\sum_{n=0}^\infty |\alpha_n^{(L+1)}(t)-\alpha_n^{(L)}(t)|^2\right | \lesssim \sum_{n=0}^\infty |\alpha_n^{(L+1)}(t)-\alpha_n^{(L)}(t)|^2+\frac 1{(L+1)^{2\xi -1}},
$$
and finally
$$ \sup_{|t|\leq \tau }\left (\sum_{n=0}^\infty |\alpha_n^{(L+1)}(t)-\alpha_n^{(L)}(t)|^2\right )^{1/2}\lesssim |\alpha_{L+1,0}| +
\frac{1}{(L+1)^{\xi -1/2}}\ .$$
 Since $\xi -1/2>1$, this easily implies that $\alpha_n^{(L)}(t)\to \alpha_n(t)$ uniformly for $t\in [-\tau ,\tau ]$ as $L\to \infty $. The convergence in the right hand side of \eqref{eq:Resonant_Equation} ,
as well as the conservation laws \eqref{eq:Conserved_equatities_N_E}, follow from
\eqref{eq:Zest}.
\end{proof}
	
	%%%%%%%%%%%%%%%%%%%%%%%%%%%%%%%%%%%%%%%%%%%%%%
	%%%%%%%%%%%%%%%%%%%%%%%%%%%%%%%%%%%%%%%%%%%%%%
	%%%%%%%%%%%%%%%%%%%%%%%%%%%%%%%%%%%%%%%%%%%%%%
	%%%%%%%%%%%%%%%%%%%%%%%%%%%%%%%%%%%%%%%%%%%%%%
	
	\subsection{An invariant manifold}
	\label{sec:invariant_manifold}	
	The central feature of the Z-Hamiltonian systems is the existence of an invariant manifold that will allow us to construct exact solutions. This is a generalization of the ones used in Refs.~\cite{Biasi} ($\parm=2$), and \cite{GG, Xu, BE} ($\parm=1$):
	\beq
		\alpha_0(t) = b(t), \qquad \alpha_{n\geq 1}(t) = f_n c(t) p(t)^{n-1},
	\label{eq:invariant_manifold}
	\eeq
	with $b,c,p \in\mathbb{C}$ being three unknowns, and $f_n$ is the same time-independent sequence as in (\ref{eq:C_nmij_equation}). 
	To guarantee finite values for the conserved quantities $N$ and $E$ in (\ref{eq:Conserved_equatities_N_E}), we work with initial conditions with $|\alpha_n|$ exponentially suppressed for large $n$. It requires that
	\beq
	|p|^2 < x_c := \frac{(\parm-1)^{\parm-1}}{\parm^\parm},
	\label{eq:xc_value}
	\eeq
	 such that the exponential growth of $f_n$ is compensated ($f_{n\gg 1} \sim x_c^{-n/2} n^{-3/4}$). For conditions of this form, the infinite-dimensional system (\ref{eq:Resonant_Equation}) is reduced to three equations
	\begin{align}
		&i \dot{p} = p \left((\parm-1) N+\frac{\bar{b} c \bar{p}}{x} + \frac{(\parm-1) F(x) }{x} b \bar{c} p-(\parm-2) \frac{1-(\parm-1) F(x)}{F(x)+1}E\right), \label{eq:pdot}\\
		&i \dot{b} = b \left(N+\frac{\parm E}{F(x)+1}\right)+\frac{\parm E F(x) }{(F(x)+1) x}c \bar{p}, \label{eq:bdot}\\
		&i \dot{c} = c \left(\parm E+(\parm+1) N\right) + \frac{(\parm-1) F(x)-2}{F(x)+1}\parm E(c-b p), \label{eq:cdot}
	\end{align}
	where the dot represents the time-derivative, and we have defined $x:=|p|^2$ and the generating function
	\beq
		F(x) = \sum_{n=1}^{\infty}f_n^2 x^n.
		\label{eq:generating_function_F}
	\eeq
	See Appendix~\ref{appendix:Function_F} for details. 
	
	Contrary to Refs.~\cite{Biasi} ($\parm=2$) and \cite{GG,BE,Xu} ($\parm=1$), we do not have an explicit expression for the generating function $F(x)$ for all $\parm\geq 1$. We use instead its differential equation and the expression $x(F)$:
	\beq
	x F' = \frac{F(1+F)}{1 - (\parm-1) F}, \qquad \text{and} \qquad x = \frac{F}{(1+F)^{\parm}},
	\label{eq:Differential_eq_F}
	\eeq
	where $F' = dF/dx$. They allow us to work with $F$ instead of $x$ for most of the analysis. Recalling the admissible range of values for $x\in[0,x_c)$ explained in (\ref{eq:xc_value}), one can see that $F(x)$ and $F'(x)$ grow in that interval but only the derivative diverges, as long as $\parm>1$,
	\beq
	F(0)=0,\quad F'(0)=1, \quad F_c:=F(x_c)= \frac{1}{\parm-1}, \quad \text{and} \quad F'(x) \underset{x\to x_c}{\to} \infty.
	\label{eq:F_particular_values}
	\eeq
	Therefore, we have the admissible range of values $F\in[0,F_c)$. The upper edge $F_c$ will be particularly important for the solutions that undergo condensation and energy cascades. As we shall see, these phenomena are associated with the approach of $F$ to $F_c$.
	
	%%%%%%%%%%%%%%%%%%%%%%%%%%%%%%%%%%%%%%%%%%%%%%%%%%%%%%%%%%%
	%%%%%%%%%%%%%%%%%%%%%%%%%%%%%%%%%%%%%%%%%%%%%%%%%%%%%%%%%%%
	
	\subsection{General strategy to obtain explicit solutions}
	\label{subsec:GeneralStrategy}
	
	The strategy to obtain solutions $\alpha_n(t)$ with initial conditions in the invariant manifold (\ref{eq:invariant_manifold}) is exposed here. We begin by noting that the modulus $|b(t)|^2$ and $|c(t)|^2$ can be written in terms of the conserved quantities $N$ and $E$ given in (\ref{eq:Conserved_equatities_N_E}) and the generating function $F(t)$ in (\ref{eq:generating_function_F}), which is actually $F(x(t))$ making an abuse of notation,
	\beq
		 |b(t)|^2 = N - \frac{E}{F(t)+1} \left(1-\frac{F(t)}{F_c}\right), \qquad \text{and} \qquad |c(t)|^2 = \frac{E}{\left(F(t)+1\right)^{\parm+1}} \left(1-\frac{F(t)}{F_c}\right).
	\label{eq:b_c_in_terms_of_x_N_E}
	\eeq 
	In combination with $x(F)$ given in (\ref{eq:Differential_eq_F}) one immediately obtains the behavior for $|\alpha_0(t)|^2 = |b(t)|^2$ and $|\alpha_{n\geq 1}(t)|^2 = f_n^2 |c(t)|^2 x(t)^{n-1}$ by just studying the evolution of $F(t)$. To do so, we derive the following equation for that function:
	\beq
		\left(\frac{dF}{dt}\right)^2 + \Veff(F) = 0,
		\label{eq:Fdot}
	\eeq
	which has the form of zero-energy trajectories of a point particle moving in the potential
	\beq
		\Veff(F) = \left(\frac{F+1}{2 \parm}\right)^2 \left(C_0+C_1 F+C_2 F^2\right)
	\label{eq:Veff}
	\eeq
	with coefficients
	\beq
		C_0 = (2 \parm E-2 \parm N+S)^2, \quad C_2 = (2 (\parm-1) \parm E+S)^2, 
		\label{eq:Coefficients_V_1}
	\eeq
	\beq
		C_1 = -8 (\parm-1) \parm^2 E^2+4 \parm^2 E (2 (\parm-1) N-S)+2 S (S-2 \parm N).
		\label{eq:Coefficients_V_2}
	\eeq
	where 
		\beq
		S := \frac{2 \parm}{F+1} \left(N + \frac{F-1}{F+1} (1-(\parm-1) F)E+(F+1)^\parm(b p \bar{c} + \bar{b} \bar{p} c)\right).
		\label{eq:S_quantity}
	\eeq
	The quantity $S$ is conserved because it comes from rewriting the Hamiltonian as $\mathcal{H} = \frac{1}{2} \left(N^2 + E S\right)$.
The equation for $\dot{F}$ comes after using the expression for $\dot{p}$ in (\ref{eq:pdot}) to write an equation for $\dot{x}$, and combine it with the conserved quantities in (\ref{eq:b_c_in_terms_of_x_N_E}) and the differential equation for $F$ in (\ref{eq:Differential_eq_F}).

	From the previous derivations, we get that the motion of $|\alpha_n(t)|$ is given by the shape of the potential $V(F)$, which is determined by the conserved quantities. As Fig.~\ref{fig:shapes_V} shows, there are only three kinds of motions: time-periodic solutions, stationary solutions, and condensation processes. It follows from the evaluation of the potential at the edges: $V(0)\geq 0$ and $V(F_c)\geq 0$. 
	The general strategy to obtain explicit solutions $\alpha_n(t)$ for these motions is the following. We first fix the initial conditions, determining the values for the conserved quantities $N$, $E$, and $S$ and with that the shape of $V(F)$. We then integrate (\ref{eq:Fdot}) to obtain a solution $F(t)$. With that, we have explicit expressions for $x(t)$, $|b(t)|$, and $|c(t)|$ through (\ref{eq:Differential_eq_F}) and (\ref{eq:b_c_in_terms_of_x_N_E}). It immediately gives the expression for the amplitude spectrum $|\alpha_{n}(t)|^2$. Finally, for the phases of $\alpha_n(t)$, we plug the decomposition $b(t) = |b(t)| e^{i \phi_b(t)}$, $c(t) = |c(t)| e^{i \phi_c(t)}$, and $p(t) = |p(t)| e^{i \phi_p(t)}$ into the equations for $\dot{b}$, $\dot{c}$, and $\dot{p}$ (\ref{eq:pdot})-(\ref{eq:cdot}) and use (\ref{eq:S_quantity}) to write $(b p \bar{c} + \bar{b} \bar{p} c)$ in terms of $F(t)$ and the conserved quantities. It gives the equations for the phases $\phi_{b,c,p}$ with a right-hand side that only depends on $t$ and can be integrated. Gathering all these elements, one obtains a solution $\alpha_n(t)$.

		\begin{figure}[h!]
		\centering
			\includegraphics[width = 5.4cm]{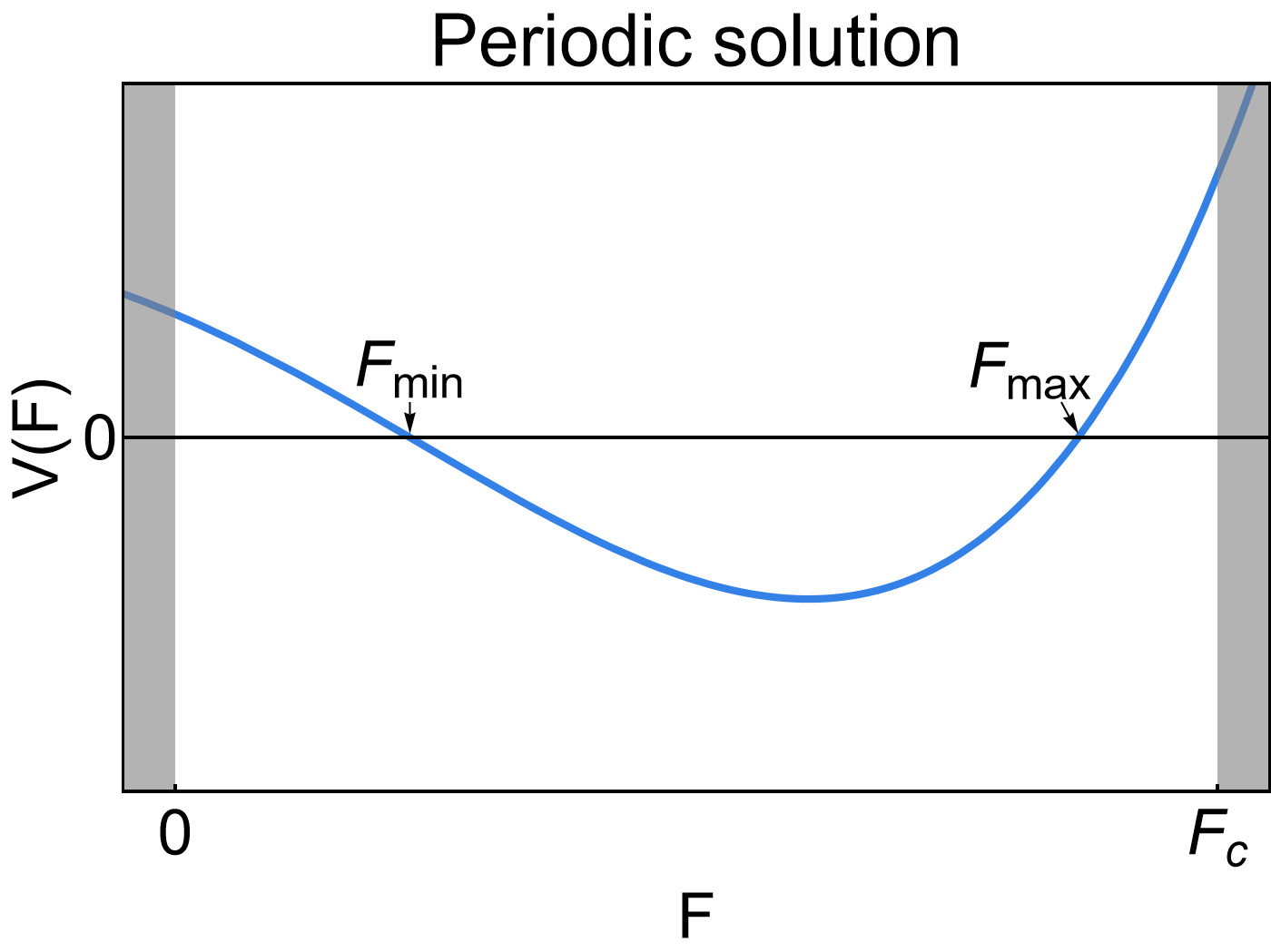}
			\includegraphics[width = 5.4cm]{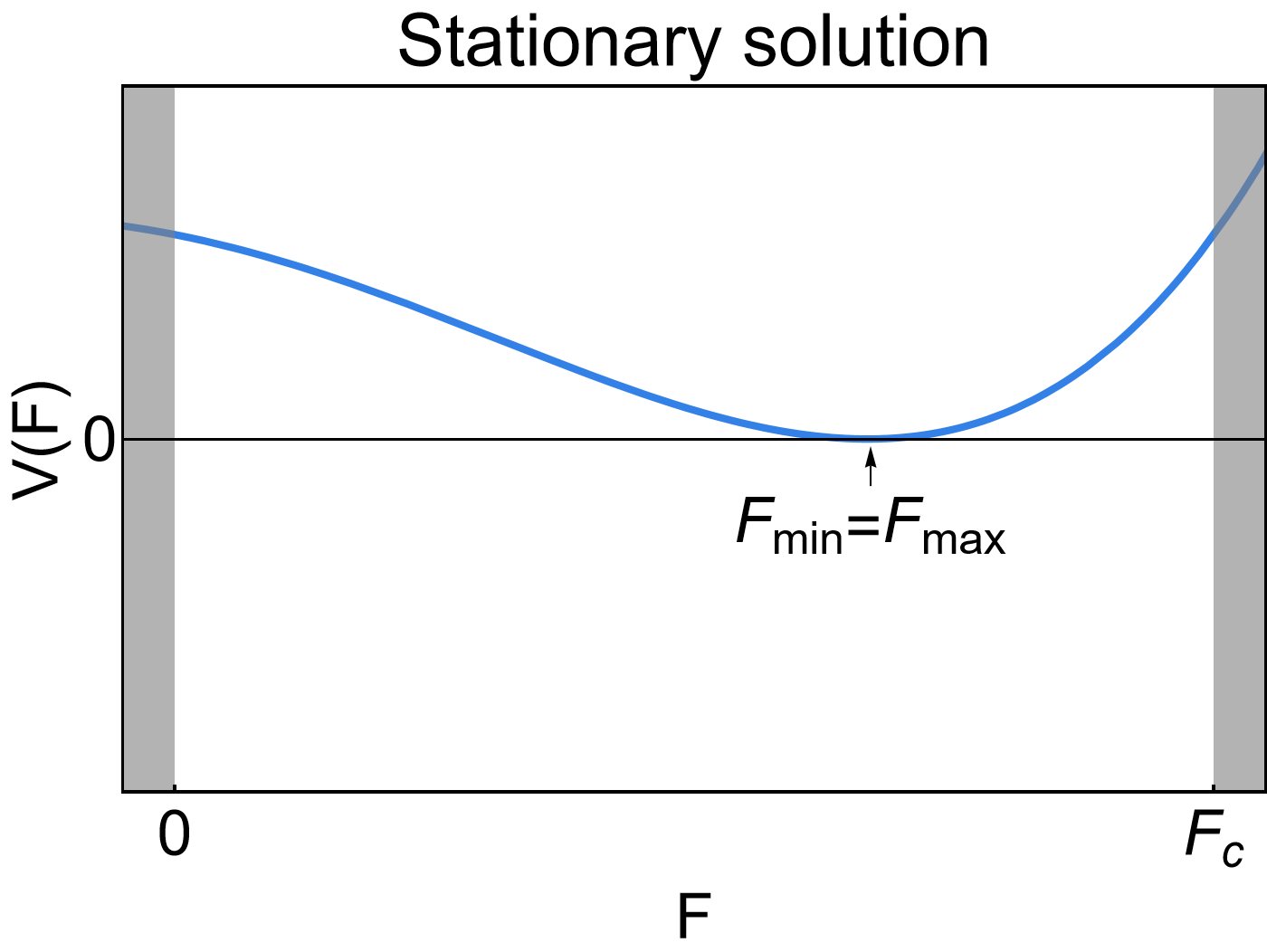}
			\includegraphics[width = 5.4cm]{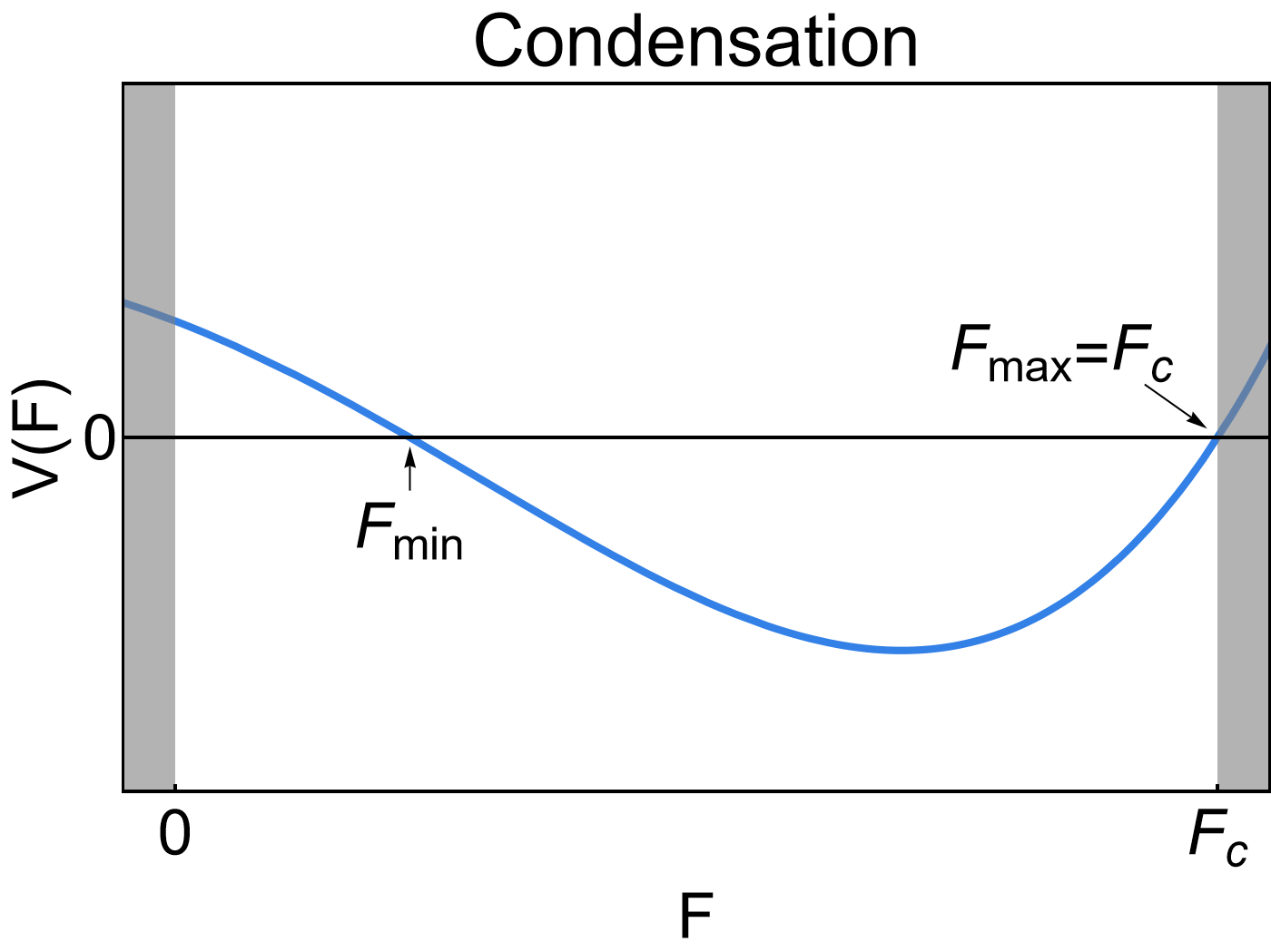}
			\caption{\small Shapes of $V(F)$ for the motions present in the invariant manifold (\ref{eq:invariant_manifold}). The shaded areas represent values of $F$ out of $[0,F_c)$ (i.e., $x$ out of $[0,x_c)$). $F_{\text{min}}$ ($F_{\text{max}}$) represents the minimum (maximum) value of $F(t)$.}
		\label{fig:shapes_V}
	\end{figure}
	
	%%%%%%%%%%%%%%%%%%%%%%%%%%%%%%%%%%%%%%%%%%%%%%
	%%%%%%%%%%%%%%%%%%%%%%%%%%%%%%%%%%%%%%%%%%%%%%
	%%%%%%%%%%%%%%%%%%%%%%%%%%%%%%%%%%%%%%%%%%%%%%
	%%%%%%%%%%%%%%%%%%%%%%%%%%%%%%%%%%%%%%%%%%%%%%
	
	\subsection{Coherent condensation and dual-cascade behavior}
	\label{sec:coherent_condensation}

	We now construct explicit solutions representing $-3/2$-cascades (section~\ref{subsec:exact_solutions_condensation}). As we shall see, they accompany the dynamical formation of a condensate - i.e., a process where the amplitude spectrum converges to the Kronecker delta distribution: $|\alpha_n|^2 \to N \delta_{0,n}$. A dual-cascade behavior and the blow-up of Sobolev norms will be discussed in section~\ref{sec:Dual cascade behavior and blow-up of Sobolev norms}, while the formation of large- and small-scale structures in position space in section~\ref{sec:Position_space}. 
	
	%%%%%%%%%%%%%%%%%%%%%%%%%%%%%%%%%%%%%%%%%%%%%%%%
	%%%%%%%%%%%%%%%%%%%%%%%%%%%%%%%%%%%%%%%%%%%%%%%%
	
	\subsubsection{Explicit solutions representing the formation of condensates}
	\label{subsec:exact_solutions_condensation}
	
	 Recalling that energy cascades require the spectrum to develop a power law: $|\alpha_{n\gg1}|^2 \sim n^{\gamma}$, they occur within the invariant manifold when $x(t)$ approaches the critical value $x_c$, as directly extracted from (\ref{eq:invariant_manifold}). Namely, when $x(t)$ approaches the upper edge of its admissible values. In terms of the variable $F$, it happens when $F(t)\to F_c$, with that value given in (\ref{eq:F_particular_values}). Consequently, all energy cascades in the manifold are part of a condensation process: $|\alpha_n|^2\to N\delta_{0,n}$, as the amplitude $|c|^2$ vanishes for any $E>0$ when $F\to F_c$, as extracted from (\ref{eq:b_c_in_terms_of_x_N_E}). Then, the problem of finding solutions that exhibit either energy cascades or condensation is the same. It reduces to finding trajectories $F(t)$ in Eq.~(\ref{eq:Fdot}) that reach the upper edge $F_c$. A standard inspection of the potential $\Veff(F)$ in that equation tells us that the desired trajectories are given by the initial conditions satisfying the relation
	\beq
	S = 2 (\parm-1) N.
	\label{eq:blow-up condition}
	\eeq
	This is a necessary and sufficient condition for the condensation to take place, excluding the trivial case $N=0$. The necessity comes from the value $\Veff(F_c)$, which only vanishes when the relation (\ref{eq:blow-up condition}) is satisfied, otherwise the potential is positive at that point and $F_c$ is never reached; see Fig.~\ref{fig:shapes_V}a-b. The sufficiency follows from the equation for $\dot{F}$ presented later in (\ref{eq:Fdot_condensation}) after substituting $S = 2 (\parm-1) N$ in $V(F)$. The condensation condition (\ref{eq:blow-up condition}) can be expressed in terms of our unknowns in the form 
	\beq
		\left|\frac{\parm F }{x} c\bar{p} + \left(1-\frac{F}{F_c}\right) \ b\right|^2 = \frac{ F (F+1)^2}{x F_c} |c|^2,
	\eeq 
	which is satisfied by
	\beq
	c = b p \frac{-\parm F+(F+1) e^{i \lambda } \sqrt{F/F_c}}{F (F-(\parm-1))}, \quad \text{with} \quad \lambda \in [0,2\pi).
	\label{eq:blow_up_initial_conditions}
	\eeq
	This is the family of initial conditions that undergo condensation. The value $F=\parm-1$ is not problematic, as it just indicates that $b(t)$ may vanish at a given time while $c(t)$ remains finite when $\parm\in(1,2)$. This family of initial conditions can be written in terms of $N$, $E$, and $\parm$ as follows (up to symmetry transformations and time-translation):
	\begin{align}
		& p(0) = \frac{(\parm E - N) }{\sqrt{\parm-1} (N + \parm E)} \left(\frac{(\parm-1)(N+\parm E)^2}{(N-\parm E)^2 + (\parm-1)(N+\parm E)^2}\right)^{\parm/2}, \label{eq:p_Initial_Data_Z-Hamiltonian}\\
		& b(0) = (N + (\parm-2)E) \left(\frac{\parm N}{(N-\parm E)^2 + (\parm-1)(N+\parm E)^2}\right)^{1/2}, \label{eq:b_Initial_Data_Z-Hamiltonian}\\
		& c(0) = 2 E (\parm N)^{1/2} (N  + \parm E)^\parm \left(\frac{\parm-1}{(N-\parm E)^2 + (\parm-1)(N+\parm E)^2}\right)^{(\parm+1)/2}. \label{eq:c_Initial_Data_Z-Hamiltonian}
	\end{align}
	These expressions guarantee that the quantities $N$ and $E$ take the values here inserted and the condensation condition $S=2(\parm-1)N$ is satisfied. With this, we prove that condensation occurs in the Z-Hamiltonian systems for any value of the ``energy per particle", $E/N$, thereby demonstrating the absence of a critical energy.

	To obtain explicit solutions for the evolution of these initial conditions, we combine the equation for $\dot{F}(t)$ in (\ref{eq:Fdot}) with the condition for condensation (\ref{eq:blow-up condition}), yielding
	\beq
	\dot{F} = 2\Omega (F+1) \sqrt{\frac{(F-F_0)}{\parm(F_0+1)} \left(1-\frac{F}{F_c}\right)}, 
	\label{eq:Fdot_condensation}
	\eeq
	with 
	\beq
	F_0 = \frac{(N-\parm
		E)^2}{(N + \parm E)^2}F_c, \qquad \text{and} \qquad \Omega = \sqrt{\frac{(N-\parm E)^2 + (\parm-1) (N+\parm E)^2}{4\parm}}. 
	\label{eq:F0_Omega}
	\eeq
	This equation is solved by
	\beq
	F(t) = \frac{\parm F_0 + \left(1-\frac{F_0}{F_c}\right) \sin ^2(\Omega t)}{\parm-\left(1-\frac{F_0}{F_c}\right) \sin ^2( \Omega t )},
	\label{eq:solution_F(t)}
	\eeq
	where we used that $F_c = 1/(\parm-1)$. We now follow the strategy to obtain solutions $\alpha_n(t)$ explained in section~\ref{subsec:GeneralStrategy}. The expressions for $x(t)$, $|b(t)|$, and $|c(t)|$ follow from (\ref{eq:Differential_eq_F}) and (\ref{eq:b_c_in_terms_of_x_N_E}), while the explicit formula for the evolution of $|\alpha_n(t)|^2$ comes from (\ref{eq:invariant_manifold}). As $t$ approaches $T = \frac{\pi}{2\Omega}$, one can see that $F(t)\to F_c$, $x(t)\to x_c$, $|b(t)|^2\to N$, and $|c(t)|^2\to 0$, indicating that the system condenses to the lowest mode in finite time:
	\beq
	|\alpha_n|^2 \underset{t\to T}{\to} N \delta_{n,0}.
	\label{eq:Kronecker-Delta}
	\eeq
	The expressions for the phases of $b$, $c$, and $p$ follow from the integration explained in section~\ref{subsec:GeneralStrategy}. For 	$N\neq \parm E$ and $N\neq (2-\parm) E$, they take the form
	\beq
	\phi_b(t) = \phi_b(0) -\frac{1}{2} (\parm E+N) t -\arctan\left(\frac{2\Omega}{(\parm-2) E+N}\tan (\Omega t )\right),
	\label{eq:phase_b}
	\eeq
	\beq
	\phi_c(t) = \frac{1}{2} ((\parm-1) \parm E-(\parm+1) N) t -\parm \arctan\left(\frac{2\Omega}{\parm E+N}\tan (\Omega t)\right),
	\label{eq:phase_c}
	\eeq
	\beq
	\phi_p(t) = \phi_p(0) + \frac{1}{2} \parm (\parm E-N) t + \arctan\left(\frac{2\Omega}{N-\parm E} \tan (\Omega t)\right)-(\parm-1) \arctan\left(\frac{2\Omega}{\parm
		E+N} \tan (\Omega t)\right),
	\label{eq:phase_p}
	\eeq
	with $\phi_p(0) = 0$ if $N<\parm E$ and $\phi_p(0) = -\pi$ if $N>\parm E$,  while $\phi_b(0) = -\pi$ if $N<(2-\parm) E$ and $\phi_b(0) = 0$ if $N>(2-\parm) E$. Other initial phases are achieved via the symmetry transformation $\alpha_n(t)\to e^{i(\lambda_1+\lambda_2 n)} \alpha_n(t)$. For $N=\parm E$ we have $F_0=0$, corresponding to the two-mode initial data: $\alpha_0(0)=b(0)$, $\alpha_{1}(0) = c(0)$, and $\alpha_{n \geq 2}(0) = 0$. In that case, the phases take the form:
	\beq
	\phi_b(t) = - \parm E t -\arctan\left(\sqrt{\frac{\parm}{\parm-1}}\tan (\Omega t )\right),
	\label{eq:phase_b_two_mode}
	\eeq
	\beq
	\phi_c(t) =  - \parm E t - \parm \arctan\left(\sqrt{\frac{\parm-1}{\parm}} \tan(\Omega t)\right),
	\eeq
	\beq
	\phi_p(t) = - \frac{\pi}{2} - (\parm-1) \arctan\left(\sqrt{\frac{\parm-1}{\parm}} \tan(\Omega t)\right).
	\eeq	
	On the other hand, $N= (2-\parm) E$ means  that the lowest mode is initially zero, $b(0)=0$, as seen in Eq.~(\ref{eq:p_Initial_Data_Z-Hamiltonian}). It leads to the phases
	\beq
	\phi_b(t) = -\frac{\pi}{2} - E t,
	\label{eq:phase_b_b=0}
	\eeq
	\beq
	\phi_c(t) =  (\parm^2 - \parm - 1) E t - \parm \arctan\left(\sqrt{\parm-1} \tan (\Omega t )\right),
	\eeq
	\beq
	\phi_p(t) =  \parm (\parm-1) E t - \arctan\left(\frac{1}{\sqrt{\parm-1}} \tan (\Omega t )\right)  - (\parm-1) \arctan\left(\sqrt{\parm-1} \tan (\Omega t )\right).
	\eeq
	The phase of $\alpha_n(t)$ immediately follows from the combination of these expressions according to the invariant manifold (\ref{eq:invariant_manifold}). The ones for $n\geq 1$ are in a straight line, indicating that our solutions represent a highly coherent condensation process.

	%%%%%%%%%%%%%%%%%%%%%%%%%%%%%%%%%%%%%%%%%%%%%%
	%%%%%%%%%%%%%%%%%%%%%%%%%%%%%%%%%%%%%%%%%%%%%%
	%%%%%%%%%%%%%%%%%%%%%%%%%%%%%%%%%%%%%%%%%%%%%%
	%%%%%%%%%%%%%%%%%%%%%%%%%%%%%%%%%%%%%%%%%%%%%%

	\subsubsection{Dual cascade behavior and blow-up of Sobolev norms}
	\label{sec:Dual cascade behavior and blow-up of Sobolev norms}
	
	We here explain that the condensation process uncovered in the previous section is accompanied by a dual cascade behavior, which leads to the separation of the quantities $N$ and $E$ in the spectrum. It is useful to recall their expressions for this discussion: $N=\sum |\alpha_n|^2$ and $E=\sum n|\alpha_n|^2$. The inverse transfer of $N$ is easily understood from the convergence of the amplitude spectrum $|\alpha_n|$ to the Kronecker delta distribution in (\ref{eq:Kronecker-Delta}). Initially, $N$ is stored at different modes since $E\neq 0$, but as the system approaches the formation of the condensate $|\alpha_{n\geq 1}|$ decays to zero while $|\alpha_0|$ grows until it stores the total amount of $N$. At the same time, the fact that $|\alpha_{n\geq1}|$ decays to zero would compromise energy conservation. This is however compensated by the  transfer of energy to high modes through a $-3/2$-cascade:
	\beq
		|\alpha_{n\gg 1}|^2\sim |c|^2 n^{-3/2} \left(\frac{x}{x_c}\right)^n \sim (T-t)^2 n^{-3/2} \left(1- C (T-t)^4\right)^n,
		\label{eq:power-law}
	\eeq
	where constant factors have been omitted and $C$ depends on $\parm$, $E$, and $N$. The conservation of $E$ is precisely in the combination of $|c|^2$ decaying to zero and the series in $E$ growing to infinity due to the development of the power law  $n^{-3/2}$. Then, the quantity $E$ is supported by higher and higher modes as $t$ approaches $T$. This behavior is quantified by the blow-up of the Sobolev norms $\xi>1/2$ at time $T$, 
	\beq
	H^{\xi>1/2} = \left(\sum_{n=0}^{\infty}(n+1)^{2\xi}|\alpha_n|^2\right)^{\frac{1}{2}} \underset{t\sim T}{\sim} \tilde{C} (T-t)^{2(1-2\xi)},
	\label{eq:blow-up_Sobolev_norms}
	\eeq
	where $\tilde{C}$ is a constant depending on $E$ and $N$. The norm $H^{1/2} = \sqrt{N+E}$ cannot blow up due to the conservation of $E$ and $N$. 
	
	The expressions in (\ref{eq:power-law}) and (\ref{eq:blow-up_Sobolev_norms}) have been calculated as follows. We first wrote the amplitude spectrum in terms of the invariant manifold $|\alpha_{n\geq 1}|^2 = f_n^2 |c|^2 x^{n-1}$. Then, we used the asymptotic behavior of function $f_{n\gg 1}^2 \sim x_c^{-n} n^{-3/2}$ given in (\ref{eq:C_nmij_equation}) to calculate the leading term of the series in (\ref{eq:blow-up_Sobolev_norms}). After, we used the expressions for $|c(t)|^2$ and $x(t)$ that come from combining (\ref{eq:b_c_in_terms_of_x_N_E}), (\ref{eq:Differential_eq_F}), and $F(t)$ in (\ref{eq:solution_F(t)}). Expanding these expressions close to $T$, we obtained $|c(t)|^2\sim (T-t)^2$ and $1-x(t)/x_c\sim (T-t)^4$ up to constant factors. Bringing all together, we obtained (\ref{eq:power-law}) and (\ref{eq:blow-up_Sobolev_norms}).

	%%%%%%%%%%%%%%%%%%%%%%%%%%%%%%%%%%%%%%%%%%%%%%
	%%%%%%%%%%%%%%%%%%%%%%%%%%%%%%%%%%%%%%%%%%%%%%
	%%%%%%%%%%%%%%%%%%%%%%%%%%%%%%%%%%%%%%%%%%%%%%
	%%%%%%%%%%%%%%%%%%%%%%%%%%%%%%%%%%%%%%%%%%%%%%

		\subsubsection{Structure formation in position space}
	\label{sec:Position_space}
		
	We here show that the condensation process we are studying results in the formation of two structures of completely different scales in position space, such as Fig.~\ref{fig:All_Type_1} illustrated previously. We will see that this is associated with the aforementioned dual cascade that led to the separation of the quantities $N$ and $E$ in the spectrum. To do so, we use the same generating function as in \cite{GG,BE,Xu} ($\parm=1$) and \cite{Biasi} ($\parm=2$):
	\beq
	u(t,\theta) = \sum_{n=0}^{\infty} \alpha_n(t) e^{i n \theta} \qquad \text{with} \qquad \theta \in [0,2\pi).
	\label{eq:u_generating_function}
	\eeq
It maps the behavior of $\alpha_n(t)$ into the evolution of a spatial profile in a one-dimensional periodic box (the circle). As we show below,  for conditions that undergo condensation, $|u(t,\theta)|^2$ progressively turns into a flat profile (a condensate) as $t$ approaches $T$, and ends up filling the domain: $|u(t,\theta)|^2\sim N$. However, a layer structure (a ``spike") is developed around $\theta  = -\phi_p(t)$. It keeps a finite amplitude: $$|u(t\to T,\theta=-\phi_p(T))|^2 \to N + \sqrt{2/\pi} E \Gamma(1/4)^2\ ,$$ but its width progressively shrinks to zero. This concentration process follows a self-similar profile when $|\theta+\phi_p(t)|\ll1$ and $(T-t)\ll 1$:
	\beq
		|u(t,\theta)|^2 \sim \left|\sqrt{N} + i \left(\frac{2}{\pi}\right)^{1/4} \Gamma(1/4) \sqrt{E} \left(1 - i C\ \frac{\theta +\phi_p(t)}{(T-t)^4}\right)^{-1/4} \right|^2.
		\label{eq:u_mod}
	\eeq
	where $\phi_p(t)$ is the phase of $p(t)$, $C=4/\left(\parm (\parm-1) E^2 N^2\right)$, and $\Gamma$ is the gamma function. 
	
	Note the dramatic separation in the scales of the structures present in position space. One covers the whole domain and the other narrows to a point. The flat profile $|u(t,\theta)|^2\sim N$ is associated with the convergence of the amplitude spectrum $\alpha_n(t)$ to the Kronecker delta distribution $|\alpha_n(t)|^2\to N \delta_{n0}$. Only the lowest mode, the $\theta$-independent one, remains, storing the total amount of $N$. However, the systems cannot fully  develop the flat profile since it has $E=0$ while the initial energy, $E\neq 0$, must be conserved. The formation of a narrowing spike (i.e., the concentration of energy into an arbitrarily small region) is the mechanism by which our systems conserve energy as a condensate forms. This is caused by the migration of the energy towards modes with large $n$, which are represented by highly oscillatory Fourier modes in position space. The strong coherence of the condensation process we are studying (the phases of $\alpha_{n\geq 1}$ are in a straight line (\ref{eq:invariant_manifold})) results in an organized superposition of Fourier modes to form a coherent structure at small scales. A disorganized relation between the phases would result instead in a background of highly oscillatory small-amplitude fluctuations, such as observed in strongly out-of-equilibrium regimes of nonlinear waves \cite{Josserand}.
	
	The previous description follows from the calculation of (\ref{eq:u_mod}). The key element is the convergence to the Kronecker delta distribution $|\alpha_n(t)|^2\to N\delta_{n0}$, the development of the power-law spectrum $|\alpha_{n\gg1}(t)|^2/|c(t)| \sim n^{-3/2}$ with $|c(t)|^2$ decaying to zero, and the  alignment in the phases. In other words, the transfer of $N$ to the lowest mode, the transfer of $E$ to arbitrarily high modes, and the strong coherence of the process. First, one writes $u(t,\theta)$ in terms of the invariant manifold (\ref{eq:invariant_manifold})
	\beq
	u(t,\theta) = b(t) + \frac{c(t)}{p(t)} \sum_{n=1}^{\infty} f_n \left(p(t) e^{i \theta}\right)^n.
	\label{eq:u_XX}
	\eeq
	The decomposition of $b(t)$, $c(t)$, and $p(t)$ in phase and amplitude yields
	\beq
	|u(t,\theta)|^2 = \left||b(t)| + e^{-i(\phi_b(t)-\phi_c(t)+\phi_p(t))}\frac{|c(t)|}{x(t)^{1/2}} \sum_{n=1}^{\infty} f_n \left(x(t)^{1/2} e^{i (\theta+\phi_p(t))}\right)^n\right|^2.
	\label{eq:XX3}
	\eeq
	Using the asymptotic behavior $f_{n\gg 1} \sim x_c^{-n/2} \left(\parm/(2\pi (\parm-1)^3)\right)^{1/4} n^{-3/4}$ and the following behaviors when $(T-t)\ll1$ (coming from the explicit expressions in Sec.~\ref{subsec:exact_solutions_condensation}):
	\begin{eqnarray}
		&|b(t)|^2\sim N
		% -(\parm-1)E^2 N (T-t)^2
		, \quad & \left(1-\frac{x(t)}{x_c}\right) \sim \frac{1}{2} \parm (\parm-1) E^2 N^2 (T-t)^{4},\\
		&|c(t)|^2 \sim (\parm-1)^2  x_c E^2 N (T-t)^2, \quad  & \phi_b(t)-\phi_c(t)+\phi_p(t)  \sim - \frac{\pi}{2}
		%- \frac{1}{2} (N + (2 - \parm) E) (T - t)}
		,
	\end{eqnarray}
	one finds that the relevant contribution to the series in Eq.~(\ref{eq:u_mod}) is
	\beq
	(T-t)\sum_{n=1}^{\infty} n^{-3/4} \left(\frac{x(t)}{x_c}\right)^{n/2} e^{i n(\theta+\phi_p(t))}.
	\label{eq:series_u}
	\eeq
	The leading term in $f_n$ is the only one that contributes because induces a divergence in the series which is compensated by the factor $(T-t)$ to give a finite value; on the other hand, subleading terms vanish. The series in (\ref{eq:series_u}) is a particular case of the polylogarithm function and only diverges when $x\to  x_c$ and the complex exponential goes to one. Then,  for $\theta\neq-\phi_p(t)$ this term decays to zero as $t$ approaches $T$, indicating that $|u(t,\theta)|^2$ goes to the flat profile $N$. However, for $|\theta+\phi_p(t)|\ll 1$ there is a competition with $(T-t)\ll 1$  of the form:
	\beq
	(T-t)\sum_{n=1}^{\infty} n^{-3/4} \left(\frac{x(t)}{x_c}\right)^{n/2} e^{i n(\theta+\phi_p(t))} \sim \Gamma(1/4)\left(\frac{1}{4} \parm (\parm-1) E^2 N^2 -i \frac{\theta +\phi_p(t)}{(T-t)^4}\right)^{-1/4}.
	\eeq
	The layer structure in (\ref{eq:u_mod}) is obtained by combining this result with (\ref{eq:XX3}).

	%%%%%%%%%%%%%%%%%%%%%%%%%%%%%%%%%%%%%%%%%%%%%%
	%%%%%%%%%%%%%%%%%%%%%%%%%%%%%%%%%%%%%%%%%%%%%%
	%%%%%%%%%%%%%%%%%%%%%%%%%%%%%%%%%%%%%%%%%%%%%%
	%%%%%%%%%%%%%%%%%%%%%%%%%%%%%%%%%%%%%%%%%%%%%%

	\subsection{Stationary states}
	\label{sec:stationary_solutions}
	
	The invariant manifold (\ref{eq:invariant_manifold}) admits three families of stationary solutions, $\alpha_n(t) = \alpha_n(0) e^{-i(\lambda+n\omega)t}$:
	\beq
	 \hspace{-2.7cm}	\text{Family 1:} \hspace{0.4cm} b(t) = b(0) e^{-iN t}, \hspace{2.cm} p(t) = p(0), \hspace{2.8cm} c(t) = - \frac{b(t)p(t)}{F_0}.
		\label{eq:Stationary_Solution_Family_1}
	\eeq
	\beq
	\hspace{-3.15cm}		\text{Family 2:}  \hspace{0.4cm} b(t) = b(0)^{-i(N + \parm E) t}, \hspace{1.3cm} p(t) = p(0) e^{-i \parm N t}, \hspace{1.7cm} c(t) = b(t)p(t).
		\label{eq:Stationary_Solution_Family_2}
	\eeq
	\beq
		\text{Family 3:} \quad b(t) = b(0) e^{-i(N + (\parm-1) E) t}, \quad p(t) = p(0) e^{-i (\parm-1)(N-2E) t}, \ \ c(t) = - b(t) p(t) \frac{1-(\parm-1)F_0}{\parm F_0},
		\label{eq:Stationary_Solution_Family_3}
	\eeq
	where $F_0 = F(|p(0)|^2)$. Note that families 1 and 3 converge to the same one in the limit $\parm\to 1$, which is still a family of stationary solutions for the cubic Szeg\H{o} equation ($\parm=1$) \cite{GG}.
	
	To obtain the previous expressions we note that stationary solutions are associated with equilibrium points $V(F)=V'(F)=0$ in the potential (\ref{eq:Fdot}), such as Fig.~\ref{fig:shapes_V} illustrated earlier. These points  are given  by the following  values of  $S$:
	\beq
	S_1 = 0, \qquad S_2= 2\parm N, \qquad \text{and} \qquad S_3 = -2 (\parm-1) (E - N),
	\label{eq:Stationary_solutions_S_conditions}
	\eeq 
	which determine the previous families of stationary solutions following the labels $1$, $2$, and $3$. The expressions in Eqs.~(\ref{eq:Stationary_Solution_Family_1})-(\ref{eq:Stationary_Solution_Family_3}) come after writing Eq.~(\ref{eq:Stationary_solutions_S_conditions}) in  terms of $b$, $c$, $p$, and $F$. The frequencies come from equations (\ref{eq:pdot})-(\ref{eq:cdot}).	

	%%%%%%%%%%%%%%%%%%%%%%%%%%%%%%%%%%%%%%%%%%%%%%%%%%
	%%%%%%%%%%%%%%%%%%%%%%%%%%%%%%%%%%%%%%%%%%%%%%%%%%
	%%%%%%%%%%%%%%%%%%%%%%%%%%%%%%%%%%%%%%%%%%%%%%%%%%
	%%%%%%%%%%%%%%%%%%%%%%%%%%%%%%%%%%%%%%%%%%%%%%%%%%
	 	 
	\subsection{The cubic Szeg\H{o} equation ($\parm=1$) and its $\alpha$- and $\beta$-deformations.}
	\label{sec:The cubic Szego equation}
	
	We here discuss how part of the results obtained for the  Z-Hamiltonian systems with $\parm>1$ connect with the well-studied cubic Szeg\H{o} equation \cite{GG} ($\parm=1$). The latter can be written as a resonant Hamiltonian system (\ref{eq:Resonant_Equation}) given by the couplings
	\beq
		C_{nmkj} = 1.
	\label{eq:Szego_eq}
	\eeq
	This equation was originally introduced in \cite{GG} in the form of a differential equation for the generating function $u(t,\theta)$ that we used in section~\ref{sec:Position_space}. Its structure is very rich, presenting an infinite tower of multi-dimensional invariant manifolds and two Lax pair structures that make the system integrable.
	
	Interestingly, the family of Z-Hamiltonian systems (\ref{eq:Resonant_Equation})-(\ref{eq:C_nmij_equation}) includes the cubic Szeg\H{o} equation for $\parm=1$. However, we do not know yet if the members with $\parm>1$ enjoy similar structures of invariant manifolds and/or Lax pairs. For the moment, we know that they retain the invariant manifold presented in (\ref{eq:invariant_manifold}). However, the dynamics in that manifold are richer for $\parm>1$ than for the cubic Szeg\H{o} equation. For the latter, $\parm=1$, the manifold contains stationary and time-periodic solutions, while for $\parm>1$ it includes the phenomenon of coherent condensation and dual cascades in addition to the other two. An inspection of the role played by parameter $s$ in our developments reveals that two families of stationary solutions ($S = 0$ and $S = -2 (E - N) (\parm-1)$) and the family of condensation processes ($S=2(\parm-1)N$) bifurcate from the same family of stationary solutions ($S=0$) of the cubic-Szeg\H{o} equation, $\parm=1$.

	In addition to the cubic Szeg\H{o} equation, there exist other Hamiltonian systems that enjoy similar analytic structures. They have a Lax pair and an infinite tower of multi-dimensional invariant manifolds. These systems are called the $\beta$-Szeg\H{o} equation \cite{BE} and the $\alpha$-Szeg\H{o} equation \cite{Xu} for their close connection to the previous model. The $\beta$-family is given by: $C_{nmkj}^{(\beta)} = 1$ when $nmkj = 0$ and $C_{nmkj}^{(\beta)}= 1-\beta $ when $nmkj=0$, with $\beta \in \mathbb{R}$, while the $\alpha$-family acts on the linear part of the Hamiltonian system in (\ref{eq:Resonant_Equation}). Both models enjoy the invariant manifold we use in (\ref{eq:invariant_manifold}) (with $\parm = 1$), and they display $0$-cascades that take an infinite time to develop a power-law spectrum. Interestingly, these extensions of the cubic-Szeg\H{o} equation can be incorporated into our family of Z-Hamiltonians retaining the invariant manifold. For instance, we can apply the $\beta$-deformation:  $C_{nmkj}^{(\parm,\beta)} =	C_{nmkj}^{(\parm)}$ when $nmkj = 0$ and	$C_{nmkj}^{(\parm,\beta)} =(1-\beta)C_{nmkj}^{(\parm)}$ when $nmkj \neq 0$, where $C_{nmkj}^{(\parm)}$ are the coefficients for the Z-Hamiltonians presented in (\ref{eq:C_nmij_equation}). It then results in a multiparameter family of Hamiltonian systems that display both $0$- and $-3/2$-cascades depending on the parameters. Their analysis follows the same procedure exposed in previous sections.
	
	%%%%%%%%%%%%%%%%%%%%%%%%%%%%%%%%%%%%%%%%%%%%%%%%%%
	%%%%%%%%%%%%%%%%%%%%%%%%%%%%%%%%%%%%%%%%%%%%%%%%%%
	%%%%%%%%%%%%%%%%%%%%%%%%%%%%%%%%%%%%%%%%%%%%%%%%%%
	%%%%%%%%%%%%%%%%%%%%%%%%%%%%%%%%%%%%%%%%%%%%%%%%%%

		\section{The Y-Hamiltonian systems: $\mathbf{-5/2}$-cascades \& structure formation}
	\label{sec:Y_Hamiltonian_systems}

	This section presents our results on $-5/2$-cascades. They follow the same structure as the study of $-3/2$-cascades in section~\ref{sec:Z_Hamiltonian_systems}: a novel family of Hamiltonian systems is presented, it admits an invariant manifold, and some solutions represent energy cascades. We call this family the {\em Y-Hamiltonian systems} and are given by the couplings
	\beq
	C_{nmkj} = \begin{cases}
		\sqrt{\left((\parm-1)n+2\right)\left((\parm-1)m+2\right)\left((\parm-1)k+2\right)\left((\parm-1)j+2\right)} \frac{f_n f_m f_k f_j}{4 f_{n+m}^2} & \text{for} \quad n m k j = 0\\
		0 & \text{otherwise}
	\end{cases}
	\label{eq:C_nmij_equation_Y}
	\eeq
	where $f_n$ is the sequence that appears in the expression for the Z-Hamiltonians in section~\ref{sec:Z_Hamiltonian_systems}. Recall that it has the form $f_n=\sqrt{A_n^{(\parm)}}$ with $A_n^{(\parm)}$ being the Fuss-Catalan numbers given in (\ref{eq:Fuss_Catalan_numbers}). This family of systems has been constructed inspired by the {\em truncated Szeg\H{o} equation} introduced in \cite{BE}, which is precisely the element $\parm=1$. For its construction, we performed a truncation of the Z-Hamiltonian systems (turning to zero all interactions that excluded the lowest mode) and realized that the modification (\ref{eq:C_nmij_equation_Y}) kept a manifold invariant. It then led to a new family of analytically tractable Hamiltonian structures.

		%%%%%%%%%%%%%%%%%%%%%%%%%%%%%%%%%%%%%%%%%%%%%%
	%%%%%%%%%%%%%%%%%%%%%%%%%%%%%%%%%%%%%%%%%%%%%%
	%%%%%%%%%%%%%%%%%%%%%%%%%%%%%%%%%%%%%%%%%%%%%%
	%%%%%%%%%%%%%%%%%%%%%%%%%%%%%%%%%%%%%%%%%%%%%%
	
	\subsection{Local well-posedness}
	\label{sec:Local_well_posedness_Y-Hamiltonian}

	 Before delving into energy cascades, we here prove that the Y-Hamiltonian systems are locally well-posed.
	
	\begin{proposition}\label{th:WPY}
Let $\xi >2$ and $(\alpha_{n,0})_{n\geq 0}$ be a sequence of complex numbers such that
\beq \sum_{n=0}^\infty (1+n)^{2\xi}\vert \alpha_{n,0}\vert ^2 \leq R<\infty \ . \label{eq:Ydata}\eeq
There exists $\tau=\tau(R)>0$ and a unique sequence $(\alpha_n(t))_{n\geq 0}$ of $C^1$ functions 
on the interval $[-\tau,\tau]$ satisfying $\alpha_n(0)=\alpha_{n,0}$, the estimate 
\beq
\sup_{|t|\leq \tau} \sum_{n=0}^\infty (1+n)^{2\xi}|\alpha_n(t)|^2<\infty ,
\label{eq:Ybounds}
\eeq
and Equation \eqref{eq:Resonant_Equation} with \eqref{eq:C_nmij_equation_Y}.
\end{proposition}
\begin{proof} We follow the lines of the proof of Proposition \ref{th:WPZ}, first solving a cutoff approximation of \eqref{eq:Resonant_Equation}.
The mass conservation law \eqref{eq:mass} still holds, providing a globally defined approximate solution $(\alpha_n^{(L)}(t))$. 
 We now claim the following key a priori estimate, for some $\tau=\tau(R)>0$ to be chosen,
\beq \sup_{|t|\leq \tau}\sum_{n=0}^L n^{2\xi} |\alpha_n^{(L)}(t)|^2 \leq 2R\ .
\label{eq:Yest}
\eeq 
 Setting again $$H^\xi _L(t):=\sum_{n=0}^L n^{2\xi}|\alpha_n^{(L)}(t)|^2\ ,$$
and using again the symmetries of $C_{nmkj}$, we estimate the derivative of $H^\xi _L(t)$ by using the  following elementary estimate, under the assumption $k=n+m$,
$$\sqrt{(1+n)(1+m)(1+k)}|n^{2\xi}+m^{2\xi}-k^{2\xi}|\lesssim \max (1+m,1+n)^\xi (1+k)^\xi \min (1+m,1+n)^{3/2}\ .$$
Indeed, assuming e.g. $m\leq n$,   we obtain
\begin{align*}
\sqrt{(1+n)(1+m)(1+k)}|n^{2\xi}+m^{2\xi}-k^{2\xi}|&\lesssim \sqrt{1+n}\sqrt{1+k} (1+m)^{3/2} (n^{2\xi -1}+k^{2\xi -1})\\
&\lesssim (1+m)^{3/2} (1+n)^\xi (1+k)^\xi \ .
\end{align*}
Consequently, using \eqref{eq:3sum} and \eqref{eq:mass}, we infer
\begin{align*}
\left |\frac{dH^\xi_L}{dt}\right |&\lesssim \sum_{\substack{0\leq j,k,m,n\leq L\\ nmjk=0, j+k=m+n}}(1+j+m)^{3/2}(1+n)^\xi (1+k)^\xi |\alpha_j^{(L)}(t)||\alpha_k^{(L)}(t)||\alpha_m^{(L)}(t)||\alpha_n^{(L)}(t)|\\
&\lesssim 1+ H^\xi _L(t)^{3/2}\ ,
\end{align*} 
where we have used, for every $\xi >2$,
$\displaystyle{\sum_{n=0}^L n^{3/2}|\alpha_n^{(L)}(t)|\leq B(\xi ) H^\xi_L(t)^{1/2}}\ .$
Since $H^\xi_L (0)$ is uniformly bounded by $R$ in view of \eqref{eq:Ydata}, we get $H^\xi _L(t)\leq 2R$ for $|t|\leq \tau (R)$ small enough.\\
The final step of the proof is achieved as in  the proof of Proposition \ref{th:WPZ}.

\end{proof}

	%%%%%%%%%%%%%%%%%%%%%%%%%%%%%%%%%%%%%%%%%%%%%%%%%%%%%%%%%%%
	%%%%%%%%%%%%%%%%%%%%%%%%%%%%%%%%%%%%%%%%%%%%%%%%%%%%%%%%%%%
	
	\subsection{An invariant manifold}
	
	The key aspect of the Y-Hamiltonian systems is the presence of the following invariant manifold:
	\beq
	\alpha_0(t) = b(t), \qquad \alpha_{n\geq 1}(t) = g_n c(t) p(t)^{n-1},
	\label{eq:invariant_manifold_Y}
	\eeq
	where $b,c,p \in\mathbb{C}$ are three dynamical variables and $g_n = f_n/\sqrt{\frac{\parm-1}{2}n+1}$, with $f_n$ being the time-independent sequence in (\ref{eq:C_nmij_equation_Y}). The condition $|p|^2< x_c = \frac{(\parm-1)^{\parm-1}}{\parm^\parm}$ is still necessary to keep initial data with an exponential suppression of high modes (recall that $f_{n\gg 1} \sim x_c^{-n/2} n^{-3/4}$). The analysis of this manifold follows the ideas presented in section~\ref{sec:Z_Hamiltonian_systems}, so we focus on the results. To do so, we write the generating function for $g_n$ in terms of the one for $f_n$:
	\beq
		G(x) = \sum_{n=1}^{\infty} g_n^2 x^n = \frac{2F_c-F(x)}{(\parm+1)F_c} F(x), \qquad \text{with} \qquad F(x) = \sum_{n=1}^{\infty} f_n^2 x^n, \qquad \text{and} \qquad F_c = F(x_c) = \frac{1}{\parm-1}.
	\eeq
	This relation allows us to keep working with $F(t)$ -i.e., $F(x(t))$ with $x= |p|^2$. Recall its range of admissible values $F\in[0,F_c)$, with $F\to F_c=1/(\parm-1)$ representing the formation of a power law due to $x(F_c)=x_c$.
	
	 The invariant manifold reduces the Y-Hamiltonian systems to three equations
		\begin{align}
		&i \dot{p} = p \left(\frac{ 2 (F+1) N + \left(\frac{F}{F_c}-2\right)E}{2 (F+1) F_c} + \frac{b \bar{c}p}{x}\frac{F}{F_c}+\frac{\bar{b} c \bar{p}}{x}\right), \label{eq:pdot_Y}\\
		&i \dot{b} = b \left(N + \frac{E}{2 (F+1)}\left(2 s + \frac{F}{F_c}\right)\right)+\frac{c\bar{p}}{x} \frac{(s+1) F}{2 (F+1)} E, \label{eq:bdot_Y}\\
		&i \dot{c} = c \left((\parm+1)N + \frac{(\parm+1)}{2 (F+1)}\left(\frac{F}{F_c}-2\right) E \right) + b p \frac{(\parm+1)^2}{2 (F+1)}E. \label{eq:cdot_Y}
	\end{align}
	The amplitudes for the dynamical variables can be written as
	\beq
	|b(t)|^2 = N - \frac{2-F(t)/F_c}{2(F(t)+1)} E, \qquad |c(t)|^2 = \frac{\parm+1}{2\left(F(t)+1\right)^{\parm+1}} E, \quad \text{and} \quad |p(t)|^2 = \frac{F(t)}{(F(t)+1)^{\parm}}.
	\label{eq:b_c_in_terms_of_x_N_E_Y}
	\eeq 
	 They reflect two important differences with the Z-Hamiltonian systems. First, the phenomenon of coherent condensation does not occur since $|c|^2$ never vanishes for $E>0$. Second, there is an upper limit for the energy $E < E_c=2sN/(s-1)$ that comes  from the first expression in Eq.~(\ref{eq:b_c_in_terms_of_x_N_E_Y}). The study of the dynamics within the invariant manifold is however still reduced to an analysis of $F(t)$. This function satisfies an equation of the form
	\beq
		\dot{F}^2 + \Veff(F) = 0,
	\label{eq:Fdot_Y}
	\eeq
	where the potential is a quartic polynomial
	\beq
	\Veff(F) = C_0+C_1 F+C_2 F^2 + C_3 F^3 + C_4 F^4
	\label{eq:Veff_Y}
	\eeq
	with coefficients provided below. They depend on the parameter $\parm$ and the conserved quantities: $N,\ E$, and $S$. The latter has the form:
	\beq
		S = \frac{E \left(\frac{F}{F_c}-2\right) \left(\frac{F}{F_c}+4 s+2\right)}{4 (F+1)^2}+\frac{N \left(\frac{F}{F_c}+2 s\right)}{F+1}+ (s+1) (F+1)^{s-1} (\bar{b}c\bar{p}+b\bar{c}p),
		\label{eq:S_Y}
	\eeq
	coming from $\mathcal{H}=\frac{1}{2}\left(N^2 + E S\right)$. The coefficients in $V(F)$ are:
	\beq
	C_0 = \frac{(E (2 \parm+1)-2 N \parm+S)^2}{(\parm+1)^2}, 
	\eeq
	\beq
	C_1 = \frac{2 (E (2 \parm+1)-2 N \parm+S) (-E (\parm-1) \parm+N (1-3 \parm)+2 S)}{(\parm+1)^2}+2E (\parm+1) (E-N), 
	\eeq
	\begin{multline}
		C_2= \frac{(E (\parm-1) \parm+N (3 \parm-1)-2 S)^2}{ (\parm+1)^2} \\ 
		-\frac{\left(E (\parm-1)^2+4 N (\parm-1)-4 S\right) (E (2 \parm+1)-2 N \parm+ S)+2E(\parm+1)^3
			(E (\parm-1)+2 N)}{2 (\parm+1)^2}, 
	\end{multline}
	\beq
	C_3 = \frac{\left(E (\parm-1)^2+4 N (\parm-1)-4 S\right) (E (\parm-1) \parm+N (3 \parm-1)-2 S)}{2 (\parm+1)^2}, 
	\eeq
	\beq
	C_4 = \frac{\left(E (\parm-1)^2+4 N (\parm-1)-4 S\right)^2}{16 (\parm+1)^2}.
	\eeq
	
	%%%%%%%%%%%%%%%%%%%%%%%%%%%%%%%%%%%%%%%%%%%%%%%%%%%%%%%%%%%
	%%%%%%%%%%%%%%%%%%%%%%%%%%%%%%%%%%%%%%%%%%%%%%%%%%%%%%%%%%%
	
	\subsection{$\mathbf{-5/2}$-cascades}

	Our study of energy cascades in the Y-Hamiltonian systems is reduced to a standard inspection of potential $V(F)$ in (\ref{eq:Fdot_Y}) to find trajectories $F(t)$ that reach the value $F_c$. First, $V(F_c)$ vanishes when $S$ takes the values
	\beq
		S_{\pm} = \frac{(\parm-1)}{\parm^2} \left( \parm (2 \parm+1)N+ \frac{-4\parm^2+\parm+3}{4}E \pm \sqrt{ \frac{(\parm+1)^3}{F_c} E (E_c - E)}\right),
	\eeq
	with $E_c = 2 \parm N/(\parm-1)$. In between, $V(F_c)$ becomes negative, opening the interval $S\in(S_-,S_+)$ where energy cascades occur in finite time (that we denote\label{key} by $T$). This is concluded by using that $V(F)$ is a quartic polynomial that grows to infinity, $V(-1)<0$, $V(0)\geq 0$, and $V(F_c)<0$. When $S=S_{\pm}$ energy cascades are displayed only for the values of $E$ where $V'(F_c)>0$ as illustrated in Fig.~\ref{fig:shapes_V_Y}.

	\begin{figure}[h!]
		\centering
			\includegraphics[width =\textwidth]{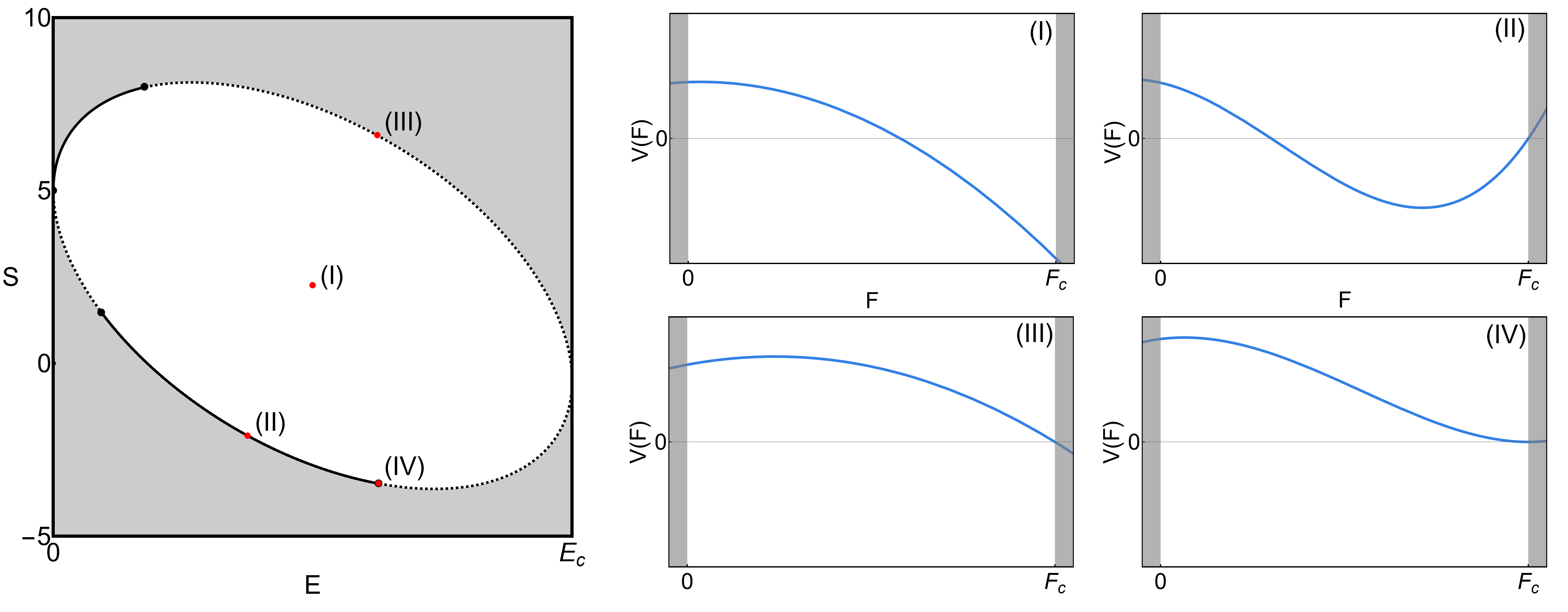}
		\caption{\small $E$-$S$ region where energy cascades take place (white area and black solid lines). In the white area $V(F_c)<0$, while $V(F_c)=0$ on the borders. Solid and dashed lines indicate $V'(F_c)>0$ and $V'(F_c)<0$, respectively. At black dots $V'(F_c)=0$. The red dots mark the values $(E,S)$ associated with the shapes of the potential presented on the right. We fixed $\parm=N=2$.}
		\label{fig:shapes_V_Y}
	\end{figure}
	
	From (\ref{eq:b_c_in_terms_of_x_N_E_Y}), we find that energy cascades lead to the spectral distribution:
	\beq
		|\alpha_0|^2 \underset{t\to T}{\longrightarrow
		} N\left(1- \frac{E}{E_c}\right), \qquad |\alpha_{n\geq 1}|^{2} \underset{t\to T}{\longrightarrow
	}  (\parm^2-1) N \frac{E}{E_c}\ g_n^2 x_c^n,
	\eeq
	Then, the asymptotic spectrum develops a power law:
	\beq
		|\alpha_{n\gg1}|^2 \underset{t\to T}{\sim
		} \frac{\parm+1}{2\parm} E \sqrt{\frac{E_c}{\pi N}} \ n^{-5/2}.
	\eeq
	The spectrum and power-law exponent clearly differ from the Kronecker delta distribution and the power exponent $-3/2$ exhibited by the cascade solutions to the Z-Hamiltonian systems. The power $n^{-5/2}$ leads to the blow-up of Sobolev norms $\xi\geq 3/4$ instead of $\xi>1/2$. For $S\in (S_-,S_+)$, i.e., $V(F_c)<0$
	\beq
	H^{\xi>3/4} = \left(\sum_{n=0}^{\infty} (n+1)^{2\xi} |\alpha_n|^2\right) \sim (T-t)^{3-4\xi}, \quad H^{3/4} = \left(\sum_{n=0}^{\infty} (n+1)^{\frac{3}{2}} |\alpha_n|^2\right) \sim \log\left(\frac{1}{T-t}\right),
	\label{eq:Sobolev_norms_Type_II_Y}
	\eeq
	while when $S=S_-$ or $S_+$ (i.e., $V(F_c)=0$) and $V'(F_c)>0$
	\beq
		H^{\xi>3/4} \sim (T-t)^{6-8\xi}, \qquad H^{3/4}  \sim \log\left(\frac{1}{T-t}\right).
	\label{eq:Sobolev_norms_Type_II_Y_2}
	\eeq
	We here used that $F(t) = F_c - \sqrt{-V(F_c)} (T-t) - V'(F_c)(T-t)^2/4 +o((T-t)^2)$.
	
	\subsection{Structure formation in position space}
	
	We now describe the position space structures that emerge from the energy cascade solutions uncovered for the Y-Hamiltonians. To do so, we use the one-dimensional periodic box 
	\beq
		u(t,\theta) = \sum_{n=0}^{\infty} \alpha_n(t) e^{i n \theta} \qquad \text{with} \qquad \theta \in [0,2\pi),
	\label{eq:u_generating_function_Y}
	\eeq
	which is useful to compare with the Z-Hamiltonian systems. 
	
	Regarding large-scale structures, our $-3/2$-cascade solutions to the Z-family resulted in the formation of a condensate: $|u(t,\theta)|^2$ approached to a flat profile (excluding a narrowing spike, which concentrated all the energy at a point). In contrast, the energy concentration process is not that efficient for our $-5/2$-cascade solutions to the Y-family. The energy remains distributed over the spatial scales, as $|c(t)|$ does not vanish. It results in $\theta$-dependent profiles for $|u(t,\theta)|^2$, as the one illustrated in Fig.~\ref{fig:All_Type_2}. We do not have an explicit expression for those profiles, but they can be approximated by a combination of polylogarithm functions expanding $g_n$ (\ref{eq:invariant_manifold_Y}) in powers of $n$.  
	
	 Regarding small-scale structures, our energy-cascade solutions for each family approach a different type of singularity as $t$ converges to $T$. For the Z-Hamiltonians, $|u(t,\theta)|^2$ gets closer and closer to a point discontinuity at $\theta=-\phi_p(T)$ (a narrowing spike collapses to a point). However, for the Y-Hamiltonians, $|u(t,\theta)|^2$ develops a pointed region around $\theta=-\phi_p(t)$, which sharpens to a cusp. The formation of this structure is led by the following profile,
	 \beq 
	 |u(T,\theta)|^2=\left|   |b(T)|e^{i (\phi_b(T)-\phi_c(T)+\phi_p(T))}+\frac{|c(T)|}{|p(T)|}\sum_{n=1}^{\infty}g_n (x(T))^{n/2}e^{in(\theta +\phi_p(T))}      \right|^2\ .
	  \eeq
	  Recall that, if $x(T)=x_c$, then $g_n (x(T))^{n/2}\sim n^{-5/4}$, which induces a cusp structure in the above trigonometric series in $\theta $.
	
	%%%%%%%%%%%%%%%%%%%%%%%%%%%%%%%%%%%%%%%%%%%%%%
	%%%%%%%%%%%%%%%%%%%%%%%%%%%%%%%%%%%%%%%%%%%%%%
	
	\subsection{Explicit solutions for $\mathbf{-5/2}$-cascades}
	
	We now provide explicit expressions for $-5/2$-cascade solutions to the Y-Hamiltonian systems with $s>1$. To do so, we choose the initial conditions 
	\beq
		\alpha_0(0)= \sqrt{\frac{\parm+1}{\parm+5}N},\qquad \alpha_1(0) = \sqrt{\frac{4N}{\parm+5}}, \qquad \text{and} \qquad \alpha_{n\geq 2}(0) = 0,
	\eeq
	which belong to the invariant manifold with: $b(0) = \alpha_0(0)$, $c(0)=\alpha_1(0)/g_1$, and $p(0)=0$. The conserved quantities satisfy the relation
	\beq
		E = \frac{4}{5+\parm} N, \qquad \text{and} \qquad S = 2N \frac{(\parm-1) (2 + \parm)}{5 + \parm},
	\eeq
	yielding a simple form for the potential in (\ref{eq:Fdot_Y}):
	\beq
	V(F) = -4 \Omega^2 \left(\frac{8}{15} + F\right) F \qquad \text{with} \qquad \Omega = \frac{\sqrt{15} (s+1)}{2 (s+5)}N.
	\eeq
	This is clearly negative for $F>0$, leading to the display of an energy cascade in finite time, as explained above. The solution is
	\beq
	F(t) = \frac{8}{15} \sinh ^2\left(\Omega t\right), 
	\label{eq:F_solution_Y}
	\eeq
	from which we obtain the time when the power-law spectrum arises (i.e., $F(T)=F_c$):
	\beq
	T = \Omega^{-1} \ \text{arcsinh} \left(\sqrt{\frac{15}{8}F_c}\right).
	\eeq
	The phases are calculated from $F(t)$ as explained in the general strategy described in section~\ref{subsec:GeneralStrategy},
	\begin{align}
		& \phi_c(t) = -\sqrt{\frac{3}{5}} (s+1) \Omega t + \frac{(s+1)}{\sqrt{7}} \text{ arctanh}\left(\sqrt{\frac{7}{15}} \tanh (\Omega t)\right),\\
		& \phi_p(t) = -\frac{\pi }{2} -\sqrt{\frac{3}{5}} (s-1) \Omega t +\frac{(s+2)}{\sqrt{7}} \text{ arctanh}\left(\sqrt{\frac{7}{15}} \tanh (\Omega t)\right),\\
		& \phi_b(t) = -2 \sqrt{\frac{3}{5}} \Omega t -\frac{1}{\sqrt{7}} \text{ arctanh}\left(\sqrt{\frac{7}{15}} \tanh (\Omega t)\right)-\arctan\left(\sqrt{\frac{3}{5}} \tanh (\Omega t)\right).
	\end{align}

	The explicit expression for $\alpha_n(t)$ directly follows from combining these phases and the expressions for $|b|^2$, $|c|^2$, and, $|p|^2$ in (\ref{eq:b_c_in_terms_of_x_N_E_Y}) with the invariant manifold (\ref{eq:invariant_manifold_Y}). The phases of $\alpha_{n\geq1}$ are all aligned, as shown by Eq.~(\ref{eq:invariant_manifold_Y}), indicating that our $-5/2$-cascade solutions are driven by highly coherent dynamics. 
	
	%%%%%%%%%%%%%%%%%%%%%%%%%%%%%%%%%%%%%%%%%%%%%%%%%%
	%%%%%%%%%%%%%%%%%%%%%%%%%%%%%%%%%%%%%%%%%%%%%%%%%%
	%%%%%%%%%%%%%%%%%%%%%%%%%%%%%%%%%%%%%%%%%%%%%%%%%%
	%%%%%%%%%%%%%%%%%%%%%%%%%%%%%%%%%%%%%%%%%%%%%%%%%%
	
	\section{Discussion and Conclusions}
	\label{sec:conclusions}

	In this work, we have delved into the  deterministic study of turbulence in Hamiltonian systems by presenting two types of energy cascades driven by highly coherent dynamics. Our $-3/2$-cascade solutions represent the phenomenon of coherent condensation in finite time. This phenomenon, first uncovered in a single Hamiltonian system in Ref.~\cite{Biasi}, is shown here to occur across an infinite number of systems. In contrast, our $-5/2$-cascade solutions have not been previously observed in the systems under consideration. They give rise to a large-scale structure other than a condensate, along with a cusp at small scales. Additionally, our results generalize previous models known to display turbulent behaviors \cite{GG,BE,Xu}, such as the cubic Szeg\H{o} equation and the truncated-Szeg\H{o} equation.

We expect that the progress made in this work leads to better comprehension of energy cascades in coherent (phase-sensitive) regimes of physical systems. In this regard, our results are very informative because illustrate, via explicit solutions, the display of energy cascades in momentum space and the formation of small- and large-scale structures in position space. Finding these phenomena in systems that directly come from well-established models of physics would be particularly interesting, such as the ones associated with the nonlinear Schr\"odinger equation. That would increase the chances of observing energy cascades and condensation processes in new regimes of cold atoms and nonlinear optics, where condensate formation is a topic of current research; see Ref.~\cite{Condensation_2020} and references therein. 

An interesting collection of questions arises from this work. The most general is to understand how common the energy cascades we have uncovered are: Are they displayed by a larger class of systems? Does this class cover physically motivated models? This is followed by predictions of energy cascades and their properties: Where does the power law come from? Can a system develop a variety of power laws or a single one? Which properties of the system would determine the power law/s? Is it possible to determine it/them from a simple analysis of the couplings $C_{nmkj}$? Another interesting collection of questions regards the phenomena that accompanies each kind of energy cascade. For the moment, the $-3/2$-cascades  and $0$-cascades uncovered in the literature only exhibit a coherent condensation phenomenon in both finite \cite{Biasi} (former) and infinite time \cite{BE,Xu} (latter). However, this is most likely associated with their observations via explicit solutions, which only provide information about specific dynamics. We expect that meticulous numerical studies will uncover a rich collection of phenomena associated with the strong transfer of energy towards high modes (see \cite{BMR,MBox,Biasi}). In this regard, our solutions for energy cascades should work invaluable for benchmarking.

Finally, at a more specific level, we wonder whether the Z-Hamiltonian and Y-Hamiltonian systems possess Lax pair structures, or if this property is exclusive to only one member of each family: the cubic Szeg\H{o} equation \cite{GG} and the truncated Szeg\H{o} equation \cite{BE}, respectively. In line with this, we also wonder whether our systems possess multi-dimensional invariant manifolds beyond the ones we have analyzed. These questions remain largely open for now.
	
	%%%%%%%%%%%%%%%%%%%%%%%%%%%%%%%%%%%%%%%%%%%%%%
	%%%%%%%%%%%%%%%%%%%%%%%%%%%%%%%%%%%%%%%%%%%%%%
	%%%%%%%%%%%%%%%%%%%%%%%%%%%%%%%%%%%%%%%%%%%%%%
	%%%%%%%%%%%%%%%%%%%%%%%%%%%%%%%%%%%%%%%%%%%%%%
	
		\noindent {\large \bf Acknowledgments:} The project that gave rise to these results received the support of a fellowship from the ”la Caixa” Foundation (ID 100010434). The fellowship code is LCF/BQ/PI24/12040029. It also received the support of the LabEx ENS-ICFP: ANR-10-LABX-0010/ANR-10-IDEX-0001-02 PSL*, and the Mar\'ia de Maeztu grant CEX2023-001318-M  funded
		by MICIU/AEI /10.13039/501100011033. The second author was partially supported by the French Agence Nationale de la Recherche under the ANR project ISAAC–ANR-23–CE40-0015-01.
	
		\noindent {\large \bf Data Availability:} Data sharing is not applicable to this article, as no datasets were generated or analyzed during the development of this work.
		
		\noindent {\large \bf Declaration:} The authors declare that they have no Conflict of interest.
	
	%%%%%%%%%%%%%%%%%%%%%%%%%%%%%%%%%%%%%%%%%%%%%%
	%%%%%%%%%%%%%%%%%%%%%%%%%%%%%%%%%%%%%%%%%%%%%%
	%%%%%%%%%%%%%%%%%%%%%%%%%%%%%%%%%%%%%%%%%%%%%%
	%%%%%%%%%%%%%%%%%%%%%%%%%%%%%%%%%%%%%%%%%%%%%%
	
	\appendix

		%%%%%%%%%%%%%%%%%%%%%%%%%%%%%%%%%%%%%%%%%%%
	%%%%%%%%%%%%%%%%%%%%%%%%%%%%%%%%%%%%%%%%%%%	
	%%%%%%%%%%%%%%%%%%%%%%%%%%%%%%%%%%%%%%%%%%%
	%%%%%%%%%%%%%%%%%%%%%%%%%%%%%%%%%%%%%%%%%%%

	\section{Details on the derivation of the system of equations (\ref{eq:pdot})-(\ref{eq:cdot})} \label{appendix:Function_F}
	
	We here provide details on the derivation of the equations for $\dot{p}$, $\dot{b}$ and $\dot{c}$ in (\ref{eq:pdot})-(\ref{eq:cdot}). 
	The process for the Y-Hamiltonian systems follows the same ideas, so this is omitted. One first uses the condition $n+m=k+j$ to write the equations in the form
	\beq
		i\frac{d\alpha_n}{dt} = \underset{n+m=k+j}{\underbrace{\sum_{m=0}^{\infty}\sum_{k=0}^{\infty}\sum_{j=0}^{\infty}}}C_{nmkj} \bar{\alpha}_m\alpha_k\alpha_j = \sum_{m=0}^{\infty}\sum_{k=0}^{n+m}C_{nmk(n+m-k)} \bar{\alpha}_m\alpha_k\alpha_{n+m-k}. 
	\eeq
	Then, the expression for the invariant manifold $\alpha_0 = b$ and $\alpha_{n\geq 1} = f_n c p^{n-1}$ in (\ref{eq:invariant_manifold}) is used.
	By gathering terms with the same prefactor $|c|^2 c$, $|c|^2 b$, $\bar{b} c c$, and $\bar{c}bb$ and recalling that $f_n = \sqrt{A_n^{(\parm)}}$, where $A_n^{(\parm)}$ are the Fuss-Calan numbers in (\ref{eq:Fuss_Catalan_numbers}), one obtains sums of the form
	\beq
		\sum_{m=1}^{\infty} A_m^{(\parm)} x^{m}, \quad \sum_{m=1}^{\infty} m A_m^{(\parm)} x^{m}, \quad \text{and} \quad \sum_{k=1}^{n+m-1} A_{n+m-k}^{(\parm)} A_{k}^{(\parm)}.
	\eeq
	The first two are the generating function $F(x)$ in (\ref{eq:generating_function_F}) and $xF'(x)$, respectively. The last one is computed by using the following property of $A_n^{(\parm)}$:
	\beq
		\sum_{k=1}^{M-1} A_{M-k}^{(\parm)} A_{k}^{(\parm)} = \frac{M-1}{\frac{\parm-1}{2}M+1} A_M^{(\parm)} \qquad \text{with }\quad M=2,3,4...
	\eeq
	This relation is the summation identity for our set of Fuss-Catalan numbers in (\ref{eq:Fuss_Catalan_numbers}) \cite{FussCatalan} after moving the terms $k=0$ and $k=M$ in the sum to the right-hand side.

	After substituting all the sums and dividing by $f_n p^{n-1}$ in the equation for $\dot{\alpha}_{n\geq 1}$, the resulting expression is linear in $n$ on both sides. Equating the coefficients gives the equations for $\dot{c}$ and $\dot{p}$. The equation for $\dot{b}$ comes from $\dot{\alpha}_0$ after following similar considerations. We further develop the expressions by using the conserved quantities written in terms of the variables $b,\ c,\ p$ to replace $|b|^2$ and $|c|^2$
	\beq
	N = |b|^2 + E \frac{F(x)}{x F'(x)}, \qquad \text{and} \qquad E = |c|^2 F'(x).
	\eeq
	 Finally, we use the differential equation given in (\ref{eq:Differential_eq_F}) to write $F'$ in terms of $F$. That equation comes from the generating function for the Fuss-Catalan numbers $B(x) = \sum_{n=0}^{\infty} A_n^{(\parm)} x^n$ by realizing that $B(x) = 1+F(x)$. That generating function satisfies the relation $B(x)-1 = x B(x)^\parm$ \cite{FussCatalan}, from which the expressions for $F'(x)$ and $x(F)$ in (\ref{eq:Differential_eq_F}) follow.

%%%%%%%%%%%%%%%%%%%%%%%%%%%%%%%%%%%%%%%%%%%%%%
%%%%%%%%%%%%%%%%%%%%%%%%%%%%%%%%%%%%%%%%%%%%%%
%%%%%%%%%%%%%%%%%%%%%%%%%%%%%%%%%%%%%%%%%%%%%%
%%%%%%%%%%%%%%%%%%%%%%%%%%%%%%%%%%%%%%%%%%%%%%

\end{document}